\documentclass[runningheads,envcountsame,envcountsect]{llncs}

\newif\iffull
\fulltrue
\newcommand{\full}[2]{\iffull#1\else#2\fi}
\newif\ifanon
\iffull\anonfalse\else\anontrue\fi

\usepackage{amsmath}
\usepackage{amssymb}
\usepackage{graphicx}
\usepackage{textcomp}
\usepackage{xspace}
\usepackage{xcolor}
\usepackage{array}
\usepackage[section]{placeins}
\usepackage{array}
\usepackage[most]{tcolorbox}
\usepackage[figuresright]{rotating}
\usepackage{longtable}

\usepackage[nospace]{cite}
\usepackage[hidelinks]{hyperref}
\usepackage[expansion=false]{microtype}
\usepackage{tcolorbox}

\AtBeginDocument{\def\doi#1{\href{https://doi.org/#1}{\nolinkurl{https://doi.org/#1}}}}

\newcommand{\mypara}[1]{\smallskip\noindent\textbf{#1.}}
\newcommand{\temph}[1]{\textbf{#1}}
\newcommand{\remove}[1]{}

\ifanon

\else

\fi

\newcommand{\calF}{\mathcal{F}}

\newcommand{\calN}{\mathbb{N}}

\usepackage{enumitem}
\setlist{nosep, leftmargin=*}

\newtcolorbox{eg}{colback=gray!5!white,colframe=black!75!black,top=2pt,bottom=2pt}
\usepackage{enumitem}
\setlist{nosep, leftmargin=*}

\title{Sovereign Grassroots Currencies: A CBDC Architecture for Credit and Monetary Policy\full{ (Full Version)}{}}
\titlerunning{Sovereign Grassroots Currencies}
\ifanon
\author{Anonymous Submission}
\authorrunning{Anonymous Submission}
\institute{}
\else
\author{Ehud Shapiro}
\authorrunning{E. Shapiro}
\institute{London School of Economics, UK, and Weizmann Institute of Science, Israel}
\fi

\iffull\else

\makeatletter
\g@addto@macro\normalsize{%
  \abovedisplayskip 1\p@ \@plus 1\p@ \@minus 1\p@
  \belowdisplayskip \abovedisplayskip
  \abovedisplayshortskip 1\p@ \@plus 1\p@
  \belowdisplayshortskip 3\p@ \@plus 1\p@ \@minus 1\p@}

\def\@spthm#1#2#3#4{\topsep 1\p@ \@plus 1\p@ \@minus 1\p@
\refstepcounter{#1}%
\@ifnextchar[{\@spythm{#1}{#2}{#3}{#4}}{\@spxthm{#1}{#2}{#3}{#4}}}
\def\@Thm#1#2#3{\topsep 1\p@ \@plus 1\p@ \@minus 1\p@
\@ifnextchar[{\@Ythm{#1}{#2}{#3}}{\@Xthm{#1}{#2}{#3}}}
\renewcommand\section{\@startsection{section}{1}{\z@}%
                       {-6\p@ \@plus -2\p@ \@minus -2\p@}%
                       {2\p@ \@plus 2\p@ \@minus 2\p@}%
                       {\normalfont\large\bfseries\boldmath
                        \rightskip=\z@ \@plus 8em\pretolerance=10000 }}
\renewcommand\subsection{\@startsection{subsection}{2}{\z@}%
                       {-6\p@ \@plus -2\p@ \@minus -2\p@}%
                       {2\p@ \@plus 2\p@ \@minus 2\p@}%
                       {\normalfont\normalsize\bfseries\boldmath
                        \rightskip=\z@ \@plus 8em\pretolerance=10000 }}
\renewcommand\subsubsection{\@startsection{subsubsection}{3}{\z@}%
                       {-9\p@ \@plus -3\p@ \@minus -3\p@}%
                       {-0.5em \@plus -0.22em \@minus -0.1em}%
                       {\normalfont\normalsize\bfseries\boldmath}}
\renewcommand\paragraph{\@startsection{paragraph}{4}{\z@}%
                       {-8\p@ \@plus -3\p@ \@minus -3\p@}%
                       {-0.5em \@plus -0.22em \@minus -0.1em}%
                       {\normalfont\normalsize\itshape}}
\def\@listI{\leftmargin\leftmargini \parsep \z@ \@plus 1\p@
  \topsep 3\p@ \@plus 1\p@ \@minus 1\p@ \partopsep 1\p@
  \itemsep 1\p@ \@plus .5\p@ \@minus .5\p@}
\let\@listi\@listI \@listi
\makeatother
\fi
\begin{document}

\maketitle

\begin{abstract}
A Central Bank Digital Currency (CBDC) is central-bank money in digital form, held by the public. Leading designs have two limitations: conversion from bank deposits into CBDC can accelerate deposit flight, requiring safeguards, and the CBDC stays outside credit creation and monetary-policy operations.

Here we present a CBDC architecture based on grassroots currencies that overcomes these limitations.  The architecture has three components: (1)  \textbf{Money:} \emph{sovereign grassroots coins}, which are digital debts of one unit of fiat currency issued by the central bank, constituting a direct CBDC; (2) \textbf{Credit and Liquidity:} \emph{non-sovereign grassroots coins}, which are digital debts of one unit of the same fiat currency, redeemable at par, that can be issued by any person, natural or legal, thus adding credit; and (3) \textbf{Interest:} grassroots bonds, sovereign and non-sovereign, adding maturity and thus  interest, standard banking instruments, and the central bank's instruments of monetary policy.

The central bank can therefore lend, absorb liquidity, set its rates and buy and sell securities in the coins and bonds the public holds, choosing the counterparties and terms of its credit operations, and without converting bank deposits into newly issued central bank money on demand.

We prove that the arbitrage-free price of any non-sovereign grassroots coin whose issuer redeems it on demand is one unit of the fiat currency.  The central bank can choose to deal with any counterparty, not just banks, and we argue that the central bank's interest rates on lending and bonds bound from above and below the corresponding interest rates of its counterparties.  Sovereign and non-sovereign grassroots coins and bonds have been implemented and tested on a small scale.

\end{abstract}

\section{Introduction}\label{sec:introduction}

A Central Bank Digital Currency (CBDC) is central-bank money in digital form, held by the public~\cite{bis2020cbdc,auer2022motives}. Leading designs have two limitations: conversion from bank deposits into CBDC can accelerate deposit flight, requiring safeguards~\cite{bis2021finstab,ecb2023stocktake,boe2023digitalpound}, and the CBDC stays outside credit creation and monetary-policy operations~\cite{mcleay2014money}.

\mypara{The contenders}
Of those in operation or in pilot we consider four, the Sand Dollar of the Bahamas~\cite{sanddollar}, the eNaira of Nigeria~\cite{cbn2021enaira}, JAM-DEX of Jamaica~\cite{boj2022jamdex} and the e-CNY of China~\cite{pboc2021ecny}; the digital euro~\cite{ecb2023stocktake,ecb2025closing} is in preparation, with a pilot possible from mid-2027 and a first issuance during 2029, and the digital pound~\cite{boe2023digitalpound,boe2026progress} is in design, with no decision taken on introducing it; and the research community has designed systems for them~\cite{senn2026sok}\full{: Project Hamilton~\cite{lovejoy2022hamilton}, Chaumian eCash~\cite{chaum2021cbdc,bis2023tourbillon}, and the privacy-preserving Platypus~\cite{wust2022platypus} and PEReDi~\cite{kiayias2022peredi}}{~\cite{lovejoy2022hamilton,chaum2021cbdc,bis2023tourbillon,wust2022platypus,kiayias2022peredi}}.  Beside them stand the private alternatives, regulated stablecoins~\cite{gorton2023taming,genius2025act,mica2023} and tokenised deposits~\cite{garratt2023singleness}.

The designs are classified by who holds the retail claim and who keeps the retail record~\cite{auer2020technology,auer2020rise}: in a direct CBDC the central bank keeps every retail balance and executes every payment; in a hybrid or intermediated CBDC intermediaries do, the central bank keeping a copy or a wholesale ledger; in an indirect CBDC the claim is on an intermediary, backed by central bank money.  Each, and each private alternative, places one record and one party inside every retail payment.

Seven central banks and the BIS have set what a retail CBDC must satisfy~\cite{bis2020cbdc}: three foundational principles\full{---do no harm to monetary and financial stability, coexistence with cash and other forms of money, and innovation and efficiency---}{ }and fourteen core features\full{, among them convertibility at par, availability offline, low or no cost to end users, instant settlement, resilience, scalability, interoperability, a robust legal framework and conformity with regulatory standards}{}.

\mypara{The two limitations}
The first limitation is the subject of the central banks' own report on financial stability~\cite{bis2021finstab}.  Deposits are converted into central bank money at will, disintermediating the banks, and a run into it costs less than a run into cash, so a systemic run may be faster and larger.\full{  How much so depends on the design of the CBDC and on the credibility of deposit insurance.  The safeguards proposed are limits on holdings and on transactions, tiered remuneration and waterfalls into bank accounts.}{}

The second limitation holds of all of them: credit creation and monetary-policy operations remain outside the CBDC, in bank money and reserves. That is a choice those designs made, not something a CBDC must be: interest-bearing central bank digital currencies have been proposed, and interest is an instrument of policy~\cite{bordo2017cbdc,davoodalhosseini2022cbdc,agur2022designing} (Section~\ref{sec:related-work}).

\mypara{Grassroots currencies}
Grassroots currencies~\cite{shapiro2024gc}, issued by natural or legal persons, are IOUs with the obligation to redeem issued grassroots coins at par against any grassroots coins held by the issuer.  Liquidity forms by mutual credit: persons who trust each other swapping their coins, with no system-wide ledger or operator.

\mypara{The proposal}
Here we present a CBDC architecture based on grassroots currencies~\cite{shapiro2024gc} that overcomes these limitations.  The architecture has three components: 
\begin{enumerate}
    \item  \textbf{Money:} \emph{sovereign grassroots coins}, which are digital debts of one unit of fiat currency issued by the central bank, constituting a direct CBDC;
    \item \textbf{Credit and Liquidity:} \emph{non-sovereign grassroots coins}, which are digital debts of one unit of the same fiat currency, redeemable at par, that can be issued by any person, natural or legal, thus adding credit;
    \item  \textbf{Interest:} grassroots bonds~\cite{shapiro2026bonds}, sovereign and non-sovereign, adding maturity and thus interest, standard banking instruments, and the central bank's instruments of monetary policy.
\end{enumerate}

The three components are realised in one grassroots social contract~\cite{shapiro2026formalising}.  Given a denomination, call it \emph{Dollar} for simplicity, every person issues their own Dollar-denominated currency under it, and a sovereign grassroots coin is simply a Dollar-denominated grassroots coin issued by the central bank of the Dollar.  There is no separate protocol for the central bank, and the only difference between it and other issuers of Dollar-denominated grassroots coins is that their coins all carry some risk of default, whereas a central bank can issue by fiat as many Dollars as needed to redeem its sovereign grassroots Dollars.

A \emph{grassroots currency} comprises the grassroots coins and the grassroots bonds of its issuer; its coins are money and its bonds are credit.

\mypara{Money}
The central bank digital currency proposed here is the sovereign grassroots coins,  a direct CBDC~\cite{auer2020technology}: they enter circulation by exchange for fiat coins and by the central bank's own payments, and leave it by redemption into fiat coins, the central bank countersigns every payment and records it in its own log, and the public holds a claim on the central bank and on no intermediary.

\iffull
They differ in the machinery only from the designs in which the central bank keeps the ledger and settles every payment~\cite{lovejoy2022hamilton,ecb2023stocktake,boe2023digitalpound}: each person keeps their own holdings, the central bank keeps its log, a payment is disseminated over the social graph~\cite{lewis2023grassroots}, and it is final when the central bank has countersigned it, with no consensus.
\fi

\iffull
A payment is countersigned by the issuer of the coins paid and recorded in the issuer's log, held by machines the issuer designates: for a person, state custodians chosen among their friends~\cite{eitan2026securing}; for an institution, its own.
\fi

\iffull
A payment needs its payer, its payee, the issuer and those machines.  So it proceeds over any network that reaches them, Bluetooth or a local network with no connection to the Internet; it is confidential to the three parties, and the custodians' persons have no access to the log their machines hold; and it bears no fee required by the architecture, an issuer being free to charge for their services.  The central bank sees only the payments in its own coins.

A person who loses their key or device recovers their identity from their identity custodians and their log, exactly, from their state custodians, and resumes without double spending~\cite{eitan2026securing}.

A coin is a signed debt instrument, analogous to a promissory note, which is legally binding also in electronic form, with  ordinary rules of contract governing the debt it states.  Whether it is negotiable, and whether the central bank may issue it~\cite{boe2026progress}, are questions for each jurisdiction.
\fi

\iffull
The coins create no credit and open no lending channel for the central bank, and they add no instrument of policy: remuneration and limits on holdings can be added, as in the other designs, and are not part of it.
\fi

\mypara{Credit and liquidity}
With the non-sovereign grassroots currencies of persons, sovereign grassroots coins are the base money of the denomination~\cite{mcleay2014money}.  Counterparties that issue coins of their own allow the central bank to issue its coins against them.  In this component the central bank issues its coins by mutual credit lines.

\iffull
Issuance that is not by conversion is Kumhof and Noone's principle~\cite{kumhof2018cbdc}: no guaranteed on-demand convertibility of bank deposits into the CBDC, and issuance only against eligible securities, principally government securities.  There the central bank buys securities its counterparty already holds; here it lends against coins its counterparty issues.
\fi

\iffull
A mutual credit line is a swap of its coins for the coins of the counterparty it chooses---a bank, a city, a firm, a cooperative or a household, any counterparty the architecture admits and its mandate allows---in the amount and at the time it chooses, with no bank between them and no account at the central bank.
\fi

\iffull
The counterparty's coins are redeemable against their holdings, and their coins in circulation are redeemable in full exactly when their holdings cover them~\cite{shapiro2024gc}.
\fi

\iffull
As the central bank's lending creates reserves today, its mutual credit lines create sovereign grassroots coins, and unlike reserves the public holds and pays them.  The sovereign grassroots coins and bonds in circulation are at most what the central bank has minted and parted with, redemption into fiat coins reducing them and no demand of a depositor increasing them; maturity turns one of its bonds into one of its coins and adds no debt.  They enter circulation only by transactions of the central bank, and thereafter circulate by payment, swap and redemption.
\fi

\iffull
A swap of coins is at par, on demand and bears no interest, so the central bank has the counterparty, the quantity and the timing of the liquidity it provides, and no rate.  Every other person creates credit independently of it.
\fi

Person-issued non-sovereign grassroots coins are redeemable at par into sovereign grassroots coins along a chain of liquidity from their issuer to the central bank, and then by it into fiat coins.  The arbitrage-free price of a coin whose issuer redeems it on demand is one unit of the fiat currency, so, while redemptions are met, the currencies of the persons and of the sovereign are one currency, with the fiat currency as its unit of account.

\iffull
Credit is created by every person, as inside money~\cite{cavalcanti1999inside,mcleay2014money} secured by social collateral~\cite{besley1995social}, which none of the designs compared allows; a bank issues its coins in place of deposits.
\fi

\iffull
The central bank holds no retail accounts and takes part in no payment but in its own grassroots currency.  Households, firms and banks accept payments in their own currencies, in the currency of the bank they are customers of, or in the sovereign grassroots currency.  Banks settle among themselves by set-off and chain redemption, and need not settle every payment between them in central bank money.
\fi

\iffull
Sovereign grassroots currencies add no on-demand operation that turns a bank deposit into newly issued central bank money: no holder of a bank deposit has a right to newly issued sovereign grassroots coins, and a bank's reserves are the bank's own.  A depositor obtains sovereign grassroots coins from a holder, which reallocates coins already in circulation, or from the central bank against cash their bank pays them, where the central bank elects to make that exchange, which changes the form of central bank money and not its quantity (Lemma~\ref{lem:conservation}).  No limit, tier or waterfall is therefore needed to bound the demand-driven creation of central bank money, which the central bank issues at its own election; whether a limit is wanted for other reasons --- concentration, transition or usability --- is a separate question, which this architecture leaves open.  Banks continue to issue and lend in their own non-sovereign grassroots currencies beside every other issuer.
\fi

\iffull
A run is a redemption run on one issuer and is met from their holdings.  The central bank can always meet presentations of its own sovereign coins by issuing fiat coins, so there is no open-ended digital run into central bank money.
\fi

\full{\mypara{Interest}}{}
\iffull
With grassroots bonds~\cite{shapiro2026bonds}, credit acquires time and a price.  A bond is a unit of its issuer's debt due at a later date.  Once the issuer's date reaches that date, its bearer may take a coin of the issuer for it.  The discount at which coins are swapped for a bond, with its maturity, determines its rate of interest.
\fi

\iffull
Every person can lend at interest, sell debt, open credit lines, post collateral, take deposits, and form the further financial instruments, each a swap of bonds or an escrow agreement.  A bank's bond takes the place of its time deposit.
\fi

\iffull
The central bank's instruments of monetary policy are the same operations, and it chooses the counterparty, amount, rate, maturity and collateral of its credit operations.  The policy rate is the rate at which it lends its coins against bonds.  Targeted liquidity is a term credit line to a city, a bank or a community, whose tranches are draws it consents to, each falling due at its maturity.  It lends against escrowed collateral, and the lender of last resort~\cite{bagehot1873lombard} is such lending at a penalty rate to an issuer whose bonds cover their liabilities while their coins do not.
\fi

\iffull
Interest is charged on the borrowing of sovereign grassroots coins and paid on the bonds the central bank issues to the counterparties it chooses; the coins bear none, as cash does not.  A remunerated CBDC pays its rate to whoever holds it, and so competes with deposits, which is one reason neither the digital euro nor the digital pound is to be remunerated~\cite{ecb2025closing,boe2023digitalpound}.
\fi

\full{}{\mypara{Interest}}
The central bank lends to a counterparty at one rate and pays another on the bonds it sells them. That counterparty will not pay more than the lending rate, since it can borrow from the central bank instead, and will not accept less than the bond rate, since it can buy a central bank bond instead. The two rates therefore bracket the rates it faces, at like maturity, collateral and risk. That bracket is the corridor~\cite{abad2025cbdc}, and the central bank chooses which counterparties it admits to it.

\iffull
The retail CBDC proposals compared add a retail form of central bank money and leave the central bank's lending, its asset purchases and the implementation of the policy rate in reserves.  In a sovereign grassroots currency they are transactions in the coins and bonds the public holds.
\fi

\begin{eg}
\mypara{A central bank with sovereign grassroots currencies}\\
A central bank that issues sovereign grassroots currencies can therefore lend, absorb liquidity, set its rates and buy and sell securities in the coins and bonds the public holds, choosing the counterparties and terms of its credit operations, and without converting bank deposits into newly issued central bank money on demand.
\end{eg}

\mypara{What we show}
We prove that the arbitrage-free price of any non-sovereign grassroots coin whose issuer redeems it on demand is one Dollar.  The central bank can choose to deal with any counterparty, not just banks, and we argue that the central bank's interest rates on lending and bonds bound from above and below the corresponding interest rates of its counterparties.  Sovereign and non-sovereign grassroots coins and bonds have been implemented and tested on a small scale (Appendix~\ref{app:run}).

We verify sovereign grassroots coins against the three principles and fourteen features.  We then compare them with the four CBDC architectures, the CBDCs in operation and in design, the system designs, and the private alternatives\full{: who is the retail holder's debtor, who keeps the retail record, who takes part in a payment, whether a payment can be made offline, who sees it, whether the central bank holds retail accounts, whether the central bank processes retail payments, who creates credit, what a payment costs, and how a lost key or device is recovered}{}.

\full{

}{}We also set sovereign grassroots coins, non-sovereign grassroots coins and grassroots bonds among the types of money proposed as a CBDC or in its place\full{---account-based CBDC, token-based CBDC, stablecoins and tokenised deposits---by whose liability each is, what backs it, its record, how it is created, whether it bears interest, and what can be done with it}{}.

\iffull
\mypara{Status}
The designs in operation have measured throughput, scale and ease of use; grassroots currencies run as a prototype \cite{lewis2023grassroots,eitan2026securing} specified by a grassroots social contract \cite{shapiro2026formalising,lewis2026volitional}, and have yet to be similarly measured.

Moreover, a digital economy of non-sovereign grassroots currencies can form without the central bank of the denomination, spontaneously or by spilling over from a neighbouring economy, and the central bank can join it once formed by issuing its own denominated grassroots coins, which are then the sovereign ones.
\fi

\section{Money: Sovereign Grassroots Coins}\label{sec:money}

A sovereign grassroots coin is a digital debt of one unit of a fiat currency, issued by the central bank of that currency.  Here we define it, describe how it enters and leaves circulation, and define its payment.  It is the money of the architecture and the central bank digital currency proposed here; the credit of the persons is Section~\ref{sec:credit}, and maturity and interest Section~\ref{sec:interest}.

\subsection{Denominated Grassroots Coins and Bonds}\label{sec:denominated}

\mypara{Persons, denominations and fiat coins}
A \emph{person} is natural or legal, with the gender-neutral pronoun ``they'': a person, a household, a merchant, a cooperative, a bank, a city, a government, a central bank.  We assume a potentially infinite set of persons $\Pi$ and consider only its finite subsets.  Let $\calF$ be a countable set of \emph{denominations}, disjoint from $\Pi$, the identifiers of the fiat currencies (USD, EUR, KES, \ldots).  A \emph{fiat coin} is a unit of a fiat currency itself: cash, or a unit of central bank money on the central bank's conventional books, held by a person it admits to hold it; a sovereign grassroots coin is not a fiat coin.  Anyone may hold cash, and reserves are held by the institutions the central bank admits, so the fiat coin a depositor is paid is cash, and their bank's reserves are the bank's own.

\mypara{Bonds and coins}
A grassroots coin is a unit of its issuer's debt, due now, backed by the issuer's offerings and by the coins and bonds they hold, and a grassroots bond is such a unit due at a later date~\cite{shapiro2024gc,shapiro2026bonds}.  The two are distinct instruments: a coin carries no date, and a bond is minted with a date later than its issuer's own.  Here each carries a denomination and its obligation is a fixed sum of it.

\begin{definition}[Denominated Grassroots Coin and Bond]\label{def:fiat-bonds}
An \temph{$f$-denominated grassroots $p$-coin}, denoted \textcent$_{f,p}$, is a unit of debt of the currency $f \in \calF$ issued by the person $p \in \Pi$, due now.  An \temph{$f$-denominated grassroots $p$-bond with maturity date $d$}, denoted \textcent$_{f,p,d}$, is such a unit due at the date $d \in \calN$.  The \temph{$f$-denominated grassroots currency of $p$} consists of the $f$-denominated $p$-coins and $p$-bonds.  We write \textcent$^k_{f,p}$ and \textcent$^k_{f,p,d}$ for multisets of $k$ of each.  Amounts are in the smallest unit of $f$, so every count is an integer.
\end{definition}

\mypara{Holdings and dates}
Each person holds a multiset of coins and bonds and keeps their own calendar; $d_p^*$ denotes the local current date of $p$, advanced by the calendar, and there is no global clock.  A bond is minted with a date later than its issuer's own, so no bond is born due.  A bond \textcent$_{f,q,d}$ is \emph{mature} iff $d \le d_q^*$, and its bearer may then take a $q$-coin for it, at their own election: maturity is judged by the date of the issuer and the dates of the holders play no part~\cite{shapiro2026bonds}.  A coin carries no date, so no advance of an issuer's date makes or unmakes one.

\mypara{The obligations}
An $f$-denominated grassroots $p$-coin is a promissory note signed by $p$, stating its issuer and denomination, and a $p$-bond states a maturity date besides.  Each carries three obligations:
\begin{enumerate}
    \item $p$ prices their offerings in their own $f$-denominated coins and accepts them in payment at those prices;
    \item $p$ prices every $f$-denominated coin and bond they hold at par against their own $f$-denominated coins, on standing offer; and
    \item an $f$-denominated grassroots coin is a debt of one unit of $f$, intended to be legally binding under the applicable law (Section~\ref{sec:payment}), and its bearer may present it to its issuer, who must redeem it at face value; a $p$-bond is the same debt due at its date, and once mature its bearer may take a $p$-coin for it.
\end{enumerate}
The issuer fulfils a presentation by an $f$-denominated coin or bond they hold, of the bearer's choice --- set-off, when the instrument taken is the bearer's own --- or by a fiat coin, performed outside the operations of the currency and legally enforceable.  A bearer may take a bond rather than a coin: a debt due later of a person they trust can be worth more to them than a debt due now of one they trust less.\full{  A presentation left unfulfilled evidences the issuer's default in $f$.  The issuer assumed the three obligations voluntarily, by minting the coin or bond, and declining a payment or a presentation is a breach of its terms, actionable at law where the obligation is enforceable (Section~\ref{sec:payment}).}{}

\mypara{The contract}
Denominated grassroots currencies operate under one grassroots social contract~\cite{shapiro2026formalising}.  A \emph{party} is a person, and \emph{Alice}, \emph{Bob} and \emph{Carol} are variables naming distinct parties.  \emph{Dollar} names one denomination, and an \emph{Alice Dollar-denominated coin} or \emph{bond}, \emph{Alice-coin} and \emph{Alice-bond} for short, is one issued by Alice.
\begin{enumerate}
\item \emph{Minting.}  Alice may decide to issue their coins, of any denomination, and their bonds, of any denomination and of a maturity later than their own date.
\item \emph{Date.}  Alice keeps their date at the current date and warrants it whenever an Alice-coin is taken for a matured Alice-bond.
\item \emph{Swap.}  Alice and Bob may jointly decide to swap Dollar-denominated coins and bonds they hold.
\item \emph{Acceptance.}  Alice may decide to accept Carol-coins in payment up to a limit of their choosing, if they hold a Carol-coin or a Carol-bond; each Carol-coin paid to Alice counts against the limit, and Alice may decide again.
\item \emph{Payment.}  Alice may decide to pay Bob Dollar-denominated coins they hold that Bob accepts.
\item \emph{Redemption.}  Alice may decide to redeem a Bob-coin they hold against any Dollar-denominated coin or bond Bob holds, and to take a Bob-coin for a matured Bob-bond they hold.
\item \emph{Pricing.}  Alice prices their offerings in Alice-coins and accepts them in payment at those prices.
\item \emph{Denomination.}  An Alice-coin and an Alice-bond are each Alice's debt of one Dollar, the coin due now and the bond at its date; Alice redeems an Alice-coin on presentation at face value, by a Dollar-denominated coin or bond the bearer accepts or by one Dollar.
\item \emph{Escrow.}  Alice may deposit Dollar-denominated coins and bonds with Carol for Bob; Carol and Bob may jointly decide to release them to Bob, and Carol and Alice to return them to Alice, as the agreed condition is met or fails.
\end{enumerate}

Clauses 1, 3, 5, 6 and 7 are the contract of grassroots currencies~\cite{shapiro2024gc} with a bond in place of a coin where the act allows one, clauses 2 and 9 are the date and escrow clauses of grassroots bonds~\cite{shapiro2026bonds}, and clauses 4 and 8 are new here.  Clause~7 is obligation~1, and clause~8 carries obligations 2 and~3.  Carol in clause~9 is the escrow agent.  A swap in which each party gives coins of their own issue opens a mutual credit line; one in which a party gives coins and the other their own bonds of larger face value maturing later is a loan.

The contract compiles to volition-guarded multiagent atomic transactions~\cite{lewis2026volitional,shapiro2026formalising}, which Definition~\ref{def:sgc-transactions} in Appendix~\ref{app:schemas} recalls from~\cite{shapiro2024gc,shapiro2026bonds} with the denomination and the acceptances added: mint, advance-date, swap, accept, pay, redeem, mature, deposit, release and return, with pay three transactions according to whether the issuer of the coins paid is the payee, the payer or a third party.  Two consequences are recalled with them: in any safe run the $f$-denominated debts of $p$ in a configuration are exactly those $p$ has minted, each a coin or a bond of one unit of $f$ (Lemma~\ref{lem:conservation}), and no transaction spans two denominations (Corollary~\ref{cor:no-cross-arbitrage}).  Appendix~\ref{app:schemas} writes the clauses as act schemas, maps each transaction to its clause, identifies the modality of each clause, and shows the contract syntactically grassroots.

\iffull Denomination requires no permission from the central bank of the denomination: a person may issue bonds of several denominations at once, as denominating a private debt in dollars requires no permission from the Federal Reserve.  Issuing instruments that circulate in payment may engage banking, electronic-money, securities, money-transmission and consumer-protection law, which is a matter of each jurisdiction.  A denominated grassroots currency therefore operates before and without the sovereign of its denomination, on the obligations of its issuers, and the sovereign may join it once it is formed (Section~\ref{sec:the-sovereign}).
\fi

\subsection{The Sovereign}\label{sec:the-sovereign}

Each fiat denomination has a sovereign issuer, or \emph{sovereign}: the central bank of the currency.  We assume a total function $\sigma: \calF \rightarrow \Pi$ that maps a denomination $f \in \calF$, say the Dollar, to its sovereign, written $f_\sigma \in \Pi$, say the central bank of the Dollar.

\begin{definition}[Sovereign Grassroots Currency]\label{def:sovereign}
The \temph{sovereign grassroots currency} of $f \in \calF$ is the $f$-denominated grassroots currency of the sovereign $f_\sigma$.  Its coins are the \temph{sovereign grassroots coins} of $f$, and its bonds the \temph{sovereign grassroots bonds} of $f$.
\end{definition}

A grassroots currency comprises the coins and the bonds of its issuer, and the central bank digital currency proposed here is the sovereign grassroots coins.  The sovereign bears the obligations of Definition~\ref{def:fiat-bonds} as every issuer does, and only it can issue a fiat coin to fulfil a presentation of its own coins.  It is otherwise an ordinary party: it mints, swaps, accepts, pays and redeems by the same transactions, the contract does not name it (Appendix~\ref{app:schemas}), and it may join a denominated grassroots currency already in operation as one more issuer.

\subsection{Issuance and Redemption}\label{sec:issuance}

\mypara{The obligation runs one way}
Redemption is the return of a coin to its issuer, and it is obligatory: the bearer presents the coin at their own election and the issuer is bound to fulfil the presentation (obligation~3).  Issuance is minting together with a swap or a payment, and both are transactions the issuer takes at their own election: no obligation requires an issuer to mint, or to part with what they have minted.  The debts of an issuer in circulation, coins and bonds together, therefore grow only by transactions the issuer consents to, and redemption only reduces them (Lemma~\ref{lem:conservation}).  Maturity moves face value from the issuer's bonds to their coins, at the bearer's election, and adds no debt: the bearer takes a coin for a bond the issuer had already minted and parted with.

\mypara{How sovereign grassroots coins enter circulation}
The central bank mints its coins and parts with them by a swap or a payment: it swaps them for fiat coins presented to it, and it pays them out to persons who accept them (Definition~\ref{def:sgc-transactions}).  It need not swap them for bank deposits on demand, and here it does not: no guaranteed on-demand convertibility of bank deposits into the central bank digital currency is one of Kumhof and Noone's four principles~\cite{kumhof2018cbdc}.  The central bank holds no retail account, and a person who holds no sovereign grassroots coin obtains one from the central bank against a fiat coin, or from a person who holds one, by payment, by swap or by redemption.  Buying and selling bonds and other securities for its coins belongs to the third component, where it is an open market operation (Section~\ref{sec:policy}).

\mypara{How they leave circulation}
On presentation the central bank fulfils by a fiat coin, which it issues, or by any $f$-denominated coin or bond it holds.  Fulfilling by a fiat coin is performed outside the contract, and its effect within the currency is the payment of the coin to its issuer, a transaction of Definition~\ref{def:sgc-transactions}, so Lemma~\ref{lem:conservation} covers it.  Redemption into fiat coins reduces the sovereign grassroots coins in circulation, and no demand for conversion increases them.  The sovereign grassroots coins and bonds in circulation are therefore bounded by what the central bank has issued and not taken back, and maturity moves face value from its bonds to its coins and adds nothing to the two together (Lemma~\ref{lem:conservation}).  What an issuer has \emph{issued} is what they have minted and parted with, coin or bond, here and below.

\mypara{A direct central bank digital currency}
Retail designs are classified by who holds the retail claim and who keeps the retail record~\cite{auer2020technology,auer2020rise}.  By that classification sovereign grassroots coins are a direct central bank digital currency: the public holds a claim on the central bank and on no intermediary, the central bank takes part in every payment in its coins, and no intermediary keeps a retail record.  The central bank keeps a log of the transactions in its coins rather than a table of balances by identity, and that log determines who holds each coin, so it serves as the retail record of the classification.  They differ in the machinery from the direct designs in which the central bank keeps every retail balance and executes every retail payment on its own ledger~\cite{lovejoy2022hamilton,ecb2023stocktake,boe2023digitalpound}: each person keeps their own holdings, the central bank keeps the record of its own coins, a payment is disseminated over the social graph~\cite{lewis2023grassroots}, and it is final when the central bank has countersigned it, with no consensus.

\iffull\mypara{What the money component does not do}
The coins create no credit and open no lending channel for the central bank, and they add no instrument of policy: remuneration and limits on holdings can be added, as in the other designs, and are not part of them.  Credit is the second component (Section~\ref{sec:credit}) and interest the third (Section~\ref{sec:interest}).
\fi

\subsection{A Payment and Its Digital Realisation}\label{sec:payment}

\iffull
The realisation is that of grassroots currencies~\cite{shapiro2024gc,lewis2023grassroots}: a coin is a digitally signed promissory note, the identity of a person is their signature key, and the issuer, a party to every payment in their coins (Definition~\ref{def:sgc-transactions}), countersigns every transaction in their own currency and may not withhold the countersignature from one the contract permits.  Every transaction is volition-guarded and multiagent atomic~\cite{lewis2026volitional,shapiro2026formalising}, and denomination restricts the bonds of a transaction to one denomination (Corollary~\ref{cor:no-cross-arbitrage}) and changes neither the transactions nor their guards.

\mypara{A payment}
A payment is a signed message from the payer to the payee, countersigned by the issuer of the coins paid and recorded in the issuer's log.  The countersignature proves holdership and prevents double spending~\cite{lewis2023grassroots}.  The log is held by machines the issuer designates: for a natural person, state custodians chosen among their friends~\cite{eitan2026securing}; for an institution, its own machines.  A payment is final when the issuer has countersigned it and the machines holding the issuer's log have recorded it: no shared ledger is updated, no block is awaited, and no third party agrees.

\mypara{Who takes part, and who sees}
A payment needs its payer, its payee, the issuer of the coins paid and those machines.  It is confidential to the three parties: the persons of the custodians have no access to the log their machines hold~\cite{eitan2026securing}, and there is no shared ledger, no operator and no intermediary.  A payment in sovereign grassroots coins therefore proceeds over any network that reaches the central bank and the machines holding its log, and not offline from them, and the central bank sees it.  The central bank does not see a payment in the coins of another issuer unless it is itself a party to it.

\mypara{Cost and throughput}
No intermediary takes part, so the architecture requires no fee, and the cost to a person is that of operating their device; an issuer is free to charge for their services.  The throughput of payments in sovereign grassroots coins is bounded by the central bank and by the machines holding its log, as in a direct central bank digital currency.

\mypara{The legal instrument}
A denominated grassroots coin is a cryptographically signed debt instrument, and its issuer is identified by their signature key.  In the United States a signature, contract or record may not be denied legal effect, validity or enforceability solely because it is in electronic form~\cite{esign2000}; in the European Union an electronic signature shall not be denied legal effect and admissibility as evidence solely on the grounds that it is in electronic form~\cite{eidas2014}; and the UNCITRAL-model acts say the same where they are adopted~\cite{uncitral2001esignatures}.  Three propositions are distinct, and only the first follows from those statutes: that the record is admissible and not denied effect for being electronic; that the obligation it states --- one unit of its fiat currency (Definition~\ref{def:fiat-bonds}) --- is an enforceable debt of its issuer, which rests on the ordinary rules of contract and on the jurisdiction; and that the instrument is a negotiable note within the bills and notes acts~\cite{boe1882,ucc3,geneva1930}, which it does not claim.  A denominated grassroots bond is analogous to a promissory note, and closer to one than an undenominated grassroots coin~\cite{shapiro2024gc}: it promises a sum certain in money at a date certain, which those acts require and the undenominated coin lacked.  Whether it is one within them is a question for each jurisdiction, open because the promise is surrounded by the other undertakings of Definition~\ref{def:fiat-bonds} and because a presentation may be fulfilled in kind.  The protocol records the provenance of each coin and prevents its double spending, serving the purpose that the rules of transfer and of proof of holdership serve for paper notes.  The UNCITRAL Model Law on Electronic Transferable Records~\cite{uncitral2017mletr} addresses this directly.  Where it is enacted, an electronic record is the functional equivalent of a transferable document or instrument if a reliable method identifies it as that record, keeps it capable of being subject to control, and retains its integrity, and the requirement of possession is met by exclusive control, identified with the person in control.  Whether the protocol's signatures, log and countersignature are such a method, and whether the Model Law is in force, are questions for each jurisdiction.  The Model Law is technology-neutral and does not alter the substantive law of the instrument.  The authority of a central bank to issue this liability is a fourth proposition, statutory and jurisdiction-specific: the Bank of England states that a digital pound would be introduced only once Parliament had passed the primary legislation for it~\cite{boe2026progress}.  The central bank's mandate we take as given.

\mypara{Recovery of a lost key or device}
A person who loses their key or their device recovers their identity from their identity custodians and their log, exactly, from their state custodians, and resumes without double spending~\cite{eitan2026securing}.  An institution recovers its log from its own machines.

\mypara{Status}
Grassroots Flash~\cite{lewis2023grassroots} is a protocol with proofs of safety, liveness and the grassroots property; recovery from the loss of a key or a machine is specified and proved in~\cite{eitan2026securing}.  The transactions of Definition~\ref{def:sgc-transactions} are abstract; a program is an executable realisation of them among parties who run it; and the secure coins of~\cite{eitan2026securing} are a third level, in which the issuer signs every transaction in their coins and state custodians hold copies of the log: that paper implements them by communicating volitional agents and proves that the implementation realises the specification at quiescence, conserves money, recovers the log exactly after the loss of a key or a machine, never double-spends after recovery, and finalises a requested payment while the issuer and a supermajority of its custodians stay live.  The transactions of the currency have a prototype implementation in GLP~\cite{shapiro2025glp,shapiro2026bonds}, with grassroots currencies running on a physical smartphone as a mini-app the Grassroots Super-App installs~\cite{shapiro2024gc}; the denominated agent is such a mini-app too, and the bond agent it is derived from keeps a network entry and cannot be installed as one.  The denominated agent implements the transactions of Definition~\ref{def:sgc-transactions} over the two instruments of Definition~\ref{def:fiat-bonds}, with the denomination as the argument after the issuer of every coin and bond.  The two extensions are separate: denomination adds the argument $f$ to every transaction of grassroots bonds and adds no transaction, while the acceptance of named currencies in payment, which the contract of grassroots currencies leaves to an extension~\cite{shapiro2024gc}, adds the three transactions accept, payout and spend (Appendix~\ref{app:schemas}).  The operations of the sovereign are the same operations, performed by the sovereign.  The denominated transactions run as a separate program, derived from the bond agent: the village market of~\cite{shapiro2026bonds} denominated in one fiat currency, with the central bank as an eighth party that mints sovereign grassroots coins and opens a mutual credit line with the community bank by a swap of coins for coins, a chain redemption that leaves a household holding a sovereign grassroots coin, which they pay to the central bank for a unit of the fiat currency, and every other transaction of the contract (Appendix~\ref{app:run}).  The central bank runs the same agent, and the contract does not name it.  Denomination itself changed no message type and no step of the protocol; the argument itself reached every line that names a coin or a bond, and the unit of the fiat currency, which the contract leaves outside itself, is handed over by the central bank's person and not by its agent.
\fi

\section{Sovereign Grassroots Coins as a Central Bank Digital Currency}\label{sec:cbdc}

\full{Seven central banks and the Bank for International Settlements have stated what a retail central bank digital currency must satisfy~\cite{bis2020cbdc}: three foundational principles --- do no harm to monetary and financial stability, coexistence with cash and other forms of money, and innovation and efficiency --- and fourteen core features in three groups.  \full{The instrument features are convertibility at par, convenience, acceptance and availability including offline, and low or no cost to end users; the system features are security, instant settlement, resilience, continuous availability, high throughput, scalability, interoperability and flexibility; the institutional features are a robust legal framework and conformity with regulatory standards.  }{}The same group has analysed the financial-stability implications of a central bank digital currency~\cite{bis2021finstab} and its system design~\cite{bis2021design}.}{Seven central banks and the Bank for International Settlements have stated what a retail central bank digital currency must satisfy~\cite{bis2020cbdc}, and the same group has analysed its financial-stability implications~\cite{bis2021finstab} and its system design~\cite{bis2021design}.}

Tables~\ref{tab:requirements} and~\ref{tab:features}\full{}{, in Appendix~\ref{app:tables},} check sovereign grassroots coins against the three principles and the fourteen features, naming for each the mechanism of Section~\ref{sec:money} it rests on and its status.  The coins are checked by themselves, before the credit of the persons and the bonds; where a later component strengthens a row, the cell names its section.  The paragraphs below add what a cell cannot hold.

\iffull\mypara{Disintermediation}
Deposits converted into central bank money leave the banks to fund themselves in wholesale markets, at higher cost and less stably, and to lend less or at higher rates~\cite{bis2021finstab}.  The mitigants proposed are limits on holdings, limits on transactions, no or tiered remuneration, and a waterfall that sweeps holdings above a limit into a bank account~\cite{bis2021finstab}; the digital euro adopts non-remuneration, a holding limit and a reverse waterfall that draws on a linked bank account~\cite{ecb2025closing}.  Each of them caps the central bank digital currency a person may hold, and each is needed because the design converts deposits into it on demand.  Here the obligation runs one way (Section~\ref{sec:issuance}): the central bank is bound to redeem its coins and is bound to issue them to no one, so the public holds what the central bank has issued and not taken back, coins and bonds together.  No operation turns a bank deposit into newly issued central bank money on demand: no holder of a bank deposit has a right to newly issued sovereign grassroots coins, and a bank's reserves are the bank's own, so a depositor obtains sovereign grassroots coins from a holder, which reallocates coins already in circulation, or from the central bank against cash their bank pays them, where the central bank elects to make that exchange --- which changes the form of central bank money and not its quantity (Lemma~\ref{lem:conservation}).  Here the central bank's issuance does what the mitigants do from outside, and no limit, tier or waterfall is added to bound the demand-driven creation of central bank money; whether one is wanted for concentration, transition or usability is a separate question, which this architecture leaves open.

\mypara{Runs}
A run into central bank money costs less than a run into cash, so a systemic run may be faster and larger; how much so depends on the design of the central bank digital currency and on the credibility of deposit insurance~\cite{bis2021finstab}.  A run reaches sovereign grassroots coins only through the holders of them, and the central bank can always fulfil a presentation of its own coins by issuing a fiat coin, so no presentation goes unmet and no run creates a debt of the central bank.  What the public can obtain of them is bounded by what the central bank has issued and not taken back, coins and bonds together, and maturity turns one of its bonds into one of its coins and creates no debt of it.  So a run on a bank is bounded by what that bank can pay its depositors, as it is today, and no run brings newly issued central bank money into being.

\mypara{Offline payment}
A payment in sovereign grassroots coins needs the central bank and the machines holding its log to be reachable (Section~\ref{sec:payment}), so it cannot be made offline from them, and the feature is not met.  The reports reach offline payment through the risk-taking of intermediaries~\cite{auer2020technology} or through value stored on the device~\cite{ecb2023stocktake}; neither is adopted here.  A payment in the coins of a person who is present proceeds over a local network with no connection to the Internet, which is the second component (Section~\ref{sec:credit}); it is not offline from that issuer either.

\mypara{Conformity with regulatory standards}
Full anonymity is not plausible for a central bank digital currency, which must meet the requirements against money laundering and the financing of terrorism~\cite{bis2020cbdc,bis2021design}.  Identities here are keypairs, and the central bank's log records every payment in its own coins and the counterparty of every transaction it takes, so that log is a complete transaction trail by protocol identity.  A key identifies a party to the protocol and not, by itself, a verified natural or legal person: binding the two, and applying the know-your-customer, anti-money-laundering and record-keeping rules to the party so identified, is a function of the deployment and its jurisdiction, as it is today.
\fi

\mypara{What the list does not ask for}
\full{The fourteen features ask what the money is and how it is paid; none asks who creates credit, and none asks whether the central bank's instruments of policy operate in the instrument the public holds.  Sovereign grassroots coins are checked against the list and leave credit creation and the operations of monetary policy where every other design leaves them, in bank money and central bank reserves.  That is the second limitation of Section~\ref{sec:introduction}, and the two components that follow answer it: the credit of the persons (Section~\ref{sec:credit}), and maturity and interest (Section~\ref{sec:interest}).}{The fourteen features ask what the money is and how it is paid; none asks who creates credit, or whether the central bank's instruments of policy operate in the instrument the public holds.  That is the second limitation of Section~\ref{sec:introduction}, which the credit of the persons (Section~\ref{sec:credit}) and maturity and interest (Section~\ref{sec:interest}) answer.}

\iffull\begin{footnotesize}
\setlength{\tabcolsep}{3pt}
\renewcommand{\arraystretch}{1.05}
\newcolumntype{M}[1]{>{\raggedright\arraybackslash}p{#1}}
\begin{longtable}{M{2.5cm}M{7.4cm}M{1.5cm}}
\caption{The three foundational principles of a retail central bank digital currency~\cite{bis2020cbdc} against sovereign grassroots coins: the mechanism each rests on and its status.  The coins are checked by themselves, before the credit of the persons (Section~\ref{sec:credit}) and the bonds (Section~\ref{sec:interest}); a cell names a later section where that component strengthens the row.  ``By construction'' is claimed for what the architecture fixes, ``argued'' for what rests on an argument under stated assumptions.}
\label{tab:requirements} \\
\hline
\textbf{Requirement}~\cite{bis2020cbdc} & \textbf{How sovereign grassroots coins meet it} & \textbf{Status} \\
\hline\hline
\endfirsthead
\hline
\textbf{Requirement} & \textbf{How sovereign grassroots coins meet it} & \textbf{Status} \\
\hline
\endhead
\hline
\multicolumn{3}{r}{\emph{continued on the next page}} \\
\endfoot
\hline
\endlastfoot
\multicolumn{3}{l}{\textbf{1.\ Foundational principles}} \\
\emph{1.a\ Do no harm: disintermediation} & No retail account, and no conversion of a bank deposit on demand~\cite{kumhof2018cbdc}: the central bank issues its coins in exchange for fiat coins and by its own payments, so no deposit becomes newly issued central bank money; with the credit component it issues them by mutual credit lines (Section~\ref{sec:credit}) & argued \\
\emph{1.b\ Do no harm: runs} & The central bank can always fulfil a presentation of its own coins by issuing a fiat coin, and what the public can obtain of them by redemption is bounded by what it has issued, so there is no open-ended digital run into central bank money & argued \\
\emph{1.c\ Do no harm: monetary policy} & The coins bear no interest, as cash does not, so they introduce no interest-rate competition with deposits, and policy is conducted as it is today; with the bonds the policy rate becomes the rate at which the central bank lends its coins (Section~\ref{sec:interest}) & argued \\
\emph{1.d\ Coexistence} & Redeemable at face value for a unit of the fiat currency by the central bank's obligation (obligation~3); cash and bank money continue; the central bank joins a non-sovereign grassroots currency already in operation as one more issuer & by construction \\
\emph{1.e\ Innovation and efficiency} & A payment needs no infrastructure beyond the parties' devices and the machines holding the issuer's log; every person may issue a currency of their own beside the central bank's (Section~\ref{sec:credit}) & by construction \\
\hline
\end{longtable}
\end{footnotesize}
\fi

\iffull\begin{footnotesize}
\setlength{\tabcolsep}{3pt}
\renewcommand{\arraystretch}{1.05}
\newcolumntype{M}[1]{>{\raggedright\arraybackslash}p{#1}}
\begin{longtable}{M{2.5cm}M{7.4cm}M{1.5cm}}
\caption{The fourteen core features of a retail central bank digital currency~\cite{bis2020cbdc} against sovereign grassroots coins, as Table~\ref{tab:requirements}.}
\label{tab:features} \\
\hline
\textbf{Requirement}~\cite{bis2020cbdc} & \textbf{How sovereign grassroots coins meet it} & \textbf{Status} \\
\hline\hline
\endfirsthead
\hline
\textbf{Requirement} & \textbf{How sovereign grassroots coins meet it} & \textbf{Status} \\
\hline
\endhead
\hline
\multicolumn{3}{r}{\emph{continued on the next page}} \\
\endfoot
\hline
\endlastfoot
\multicolumn{3}{l}{\textbf{2.\ Instrument features}} \\
\emph{2.a\ Convertible at par} & Redeemable by the central bank for one unit of the fiat currency, which it issues (obligation~3) & by construction \\
\emph{2.b\ Convenient} & A payment is a signed message countersigned by the central bank, on the parties' own devices; no account beyond a keypair; ease of use unmeasured & specified \\
\emph{2.c\ Accepted and available, offline included} & Acceptance by the central bank is enforced by the contract; a payment needs the payer, the payee, the central bank and the machines holding its log, so it proceeds over any network that reaches them, and not offline from the central bank & not met \\
\emph{2.d\ Low or no cost to end users} & No intermediary, so no fee is required by the architecture; what a person pays is the cost of operating their device, and what an issuer charges for its services, neither measured & by construction; cost unmeasured \\
\hline
\multicolumn{3}{l}{\textbf{3.\ System features}} \\
\emph{3.a\ Secure} & The central bank countersigns every payment in its coins, proving holdership and excluding double spending; a holder who loses their key or device recovers their identity from their identity custodians and their sovereign grassroots coins from the central bank's log~\cite{eitan2026securing}; these are properties of the protocol, and the resistance of a deployment to attack is not established by them & proved; to be deployed \\
\emph{3.b\ Instant settlement} & Final when the central bank has countersigned it and the machines holding its log have recorded it: no consensus and no intermediary; the time this takes is that of the network and the machines & by construction; latency unmeasured \\
\emph{3.c\ Resilient} & Payments in sovereign grassroots coins need the central bank and the machines holding its log, and no other component; a holder's loss of key or device is remedied from their custodians~\cite{eitan2026securing} & specified \\
\emph{3.d\ Continuously available} & No opening hours and no maintenance window; the availability of the coins is that of the central bank and of the machines holding its log & specified \\
\emph{3.e\ High throughput, scalable} & The central bank and the machines holding its log bound the throughput of payments in its coins, as in a direct central bank digital currency; the currencies of the persons add throughput of their own (Section~\ref{sec:credit}) & unmeasured \\
\emph{3.f\ Interoperable} & Meets the existing systems at redemption, into central bank money or a fiat coin through the banking tools of the currency; exchange across denominations by escrow (Appendix~\ref{app:fx}) & specified; to be deployed \\
\emph{3.g\ Flexible and adaptable} & Instruments are swaps and escrow agreements under conditions; a new instrument is a new condition, with no change to the contract (Section~\ref{sec:interest}) & by construction; to be deployed \\
\hline
\multicolumn{3}{l}{\textbf{4.\ Institutional features}} \\
\emph{4.a\ Robust legal framework} & Three distinct propositions: the signed record is not denied legal effect for being electronic~\cite{esign2000,eidas2014}; the obligation it states is an enforceable debt of the central bank of one unit of the fiat currency, which rests on the ordinary rules of contract and on the jurisdiction; and the central bank's authority to issue this liability is statutory, the digital pound to be introduced only once Parliament has passed the primary legislation~\cite{boe2026progress}.  The mandate this paper takes as given & argued \\
\emph{4.b\ Conformity with regulatory standards} & The central bank's log records every payment in its coins and the counterparty of every transaction it takes, a complete trail by protocol identity, which supports know-your-customer, anti-money-laundering and record-keeping requirements; binding a key to a verified legal person, and conformity itself, are the deployment's and the jurisdiction's & argued \\
\end{longtable}
\end{footnotesize}
\fi

\section{Credit and Liquidity: Non-Sovereign Grassroots Coins}\label{sec:credit}

Every person, natural or legal, issues an $f$-denominated grassroots currency of their own (Definition~\ref{def:fiat-bonds}), on the same three obligations the central bank bears.  Here we add them to the money of Section~\ref{sec:money}: what their credit is and how it is created, how the central bank issues its coins against theirs, why a coin of theirs is worth one unit of the fiat currency, and what a redemption run on one of them reaches.  With them, sovereign grassroots coins are the base money of the denomination.

\subsection{Every Person an Issuer}\label{sec:every-person}

\iffull\mypara{Credit by mutual credit}
Two persons open a mutual credit line by swapping coins of their own issue: each thereafter holds coins of the other, and each may pay the other in the other's coins or present them for redemption.  The credit so created is inside money~\cite{cavalcanti1999inside,mcleay2014money}, secured by social collateral~\cite{besley1995social}: the borrower's standing in their community, the acceptance of their currency, and the mutual credit lines they hold.  Default costs the issuer that standing --- an unfulfilled presentation evidences it, holders redeem what the issuer still holds, and no one extends them new credit~\cite{besley1995social,shapiro2024gc}.  Credit is thus created by every person, which none of the designs compared allows, and a bank's coins play the economic role of its demand deposits.

\mypara{Who pays in what}
Households, firms and banks accept payments in their own currencies, in the currency of the bank they are customers of, or in the sovereign grassroots currency (Definition~\ref{def:sgc-transactions}).  The central bank holds no retail account and takes part in no payment but in its own coins, to every one of which it is a party (Definition~\ref{def:sgc-transactions}).  Banks settle among themselves by set-off and by chain redemption (Section~\ref{sec:peg}), and need no central bank money to do so.

\mypara{A payment in the coins of a person}
The payment of Section~\ref{sec:payment} is the same transaction with another issuer: signed by the payer, countersigned by the issuer of the coins paid, recorded in the issuer's log by the machines holding it.  It needs its payer, its payee, the issuer and those machines, so it proceeds over any network that reaches them --- Bluetooth, or a local network with no connection to the Internet --- and it is confidential to the three parties.  No component is shared by all issuers, so none bounds the throughput of the whole, which grows with the issuers and their custodians.  The central bank does not see these payments unless it is itself a party to one.
\fi

\subsection{The Central Bank's Mutual Credit Lines}\label{sec:credit-lines}

\mypara{Issuance against the counterparty's coins}
A counterparty who issues coins of their own gives the central bank something to lend against.  Here the central bank issues its coins by mutual credit lines, and guarantees no conversion of bank deposits into them on demand.  A mutual credit line is a swap of its coins for the coins of the counterparty it chooses --- a bank, a city, a firm, a cooperative or a household, any counterparty the architecture admits and its mandate allows --- in the amount and at the time it chooses, with no bank between them and no account at the central bank.  The counterparty's coins are redeemable against their holdings, and their coins in circulation are redeemable in full exactly when their holdings cover them (Section~\ref{sec:runs})~\cite{shapiro2024gc}.

\iffull\mypara{The precedent, and the difference}
Issuance that is not by conversion is Kumhof and Noone's principle~\cite{kumhof2018cbdc}: no guaranteed on-demand convertibility of bank deposits into the central bank digital currency, and issuance only against eligible securities, principally government securities.  There the central bank buys securities its counterparty already holds; here it lends against coins its counterparty issues.  The counterparty of an asset purchase is whoever holds the security; the counterparty of a mutual credit line is whoever the central bank chooses.  We adopt that principle and not their separation of the currency from reserves: a sovereign grassroots coin is redeemed into a fiat coin on presentation, without restriction (Definition~\ref{def:fiat-bonds}), and the quantity is governed by issuance rather than by a bar on convertibility.  A mutual credit line is at par, on demand and without interest; the term credit line of Section~\ref{sec:instruments}, an escrowed facility whose draws carry maturities and interest, belongs to the third component~\cite{shapiro2026bonds}.
\fi

\mypara{Creation and the quantity of base money}
As the central bank's lending creates reserves today, its mutual credit lines create sovereign grassroots coins, and unlike reserves the public holds and pays them.  The sovereign grassroots coins and bonds in circulation are at most what the central bank has issued and not taken back (Section~\ref{sec:issuance}): they enter circulation only by transactions of the central bank, and thereafter circulate by payment, swap and redemption, redemption into fiat coins reducing them and no demand of a depositor increasing them.  Maturity moves face value from its bonds to its coins and adds nothing to the two together.

\iffull
\mypara{Quantity and allocation, and no rate}
\fi
A swap of coins is at par, on demand and bears no interest, so what the central bank sets by a mutual credit line is the counterparty, the quantity and the timing of the liquidity it provides, and not a rate.  The rate arrives with maturity (Section~\ref{sec:interest}).  Every other person creates credit independently of the central bank, by their own mutual credit lines.

\mypara{Disintermediation and runs}
\iffull Sovereign grassroots currencies add no on-demand operation that turns a bank deposit into newly issued central bank money: no holder of a bank deposit has a right to newly issued sovereign grassroots coins, and a bank's reserves are the bank's own.  A depositor obtains sovereign grassroots coins from a holder, which reallocates coins already in circulation, or from the central bank against cash their bank pays them, where the central bank elects to make that exchange, which changes the form of central bank money and not its quantity (Lemma~\ref{lem:conservation}).  No limit, tier or waterfall is therefore needed to bound the demand-driven creation of central bank money, which the central bank issues at its own election; whether a limit is wanted for other reasons --- concentration, transition or usability --- is a separate question, which this architecture leaves open\full{; the Bank of England has measured what a limit is worth where deposits are converted on demand~\cite{boe2025limits}}{}.  Banks continue to issue and to lend in their own grassroots currencies beside every other issuer.  A run is a redemption run on one issuer and is met from their holdings (Section~\ref{sec:runs}).  The central bank can always meet presentations of its own coins by issuing fiat coins, and what the public can obtain of them by redemption is bounded by what it has issued, so there is no open-ended digital run into central bank money.\fi

\subsection{The Peg and the Price of a Coin}\label{sec:peg}

Redemption reaches across chains of holdings, coin by coin.

\begin{proposition}[Chain Redemption~\cite{shapiro2024gc,shapiro2026bonds}]\label{prop:chain-redemption}
Let $p_0, p_1, \ldots, p_k$ be a sequence of distinct persons with $k \ge 1$ such that each $p_i$, $0 \le i < k$, holds an $f$-denominated $p_{i+1}$-coin.  Then from that configuration there is a sequence of $k-1$ redemptions by $p_0$ whose effect is that $p_0$ holds an $f$-denominated $p_k$-coin.
\end{proposition}

\begin{definition}[$f$-Holding Graph, Sink, Sovereign-Connected]\label{def:holding-graph}\label{def:sovereign-connected}
Given a configuration, the \temph{$f$-holding graph} is the directed multigraph with the persons as vertices and an edge $u \rightarrow v$ for each $f$-denominated $v$-coin and each $f$-denominated $v$-bond that person $u \ne v$ holds; its \temph{sink} is the sovereign $f_\sigma$; and the \temph{$f$-coin holding graph} is its subgraph of the edges of coins.  A person $p$ is \temph{sovereign-connected in $f$} if the $f$-coin holding graph has a directed path from $p$ to $f_\sigma$.
\end{definition}

The out-degree and in-degree of $p$ in the $f$-holding graph are their assets $a_p$ and liabilities $\ell_p$ in $f$, counted in coins and bonds~\cite{shapiro2024gc,shapiro2026bonds}; a coin or bond an issuer holds of their own carries no value and no edge.  A directed path from $p$ to $q$ in the $f$-coin holding graph is a chain of mutual liquidity: by Proposition~\ref{prop:chain-redemption}, after a sequence of redemptions along it $p$ holds an $f$-denominated $q$-coin, and the coins along it are consumed.  In particular $f_\sigma$ is sovereign-connected in $f$, by the empty path.

\begin{proposition}[Peg Propagation]\label{prop:peg}
Let $q$ hold an $f$-denominated $p$-coin, where $p$ is sovereign-connected in $f$.  Then there is a sequence of redemptions by $q$, in the denomination $f$, whose effect is that $q$ holds a sovereign grassroots coin.
\end{proposition}
\begin{proof}
$q$ holding an $f$-denominated $p$-coin is an edge $q \rightarrow p$ of the $f$-coin holding graph, and $p$ being sovereign-connected in $f$ is a directed path from $p$ to the sovereign, so the graph has a directed path from $q$ to $f_\sigma$, and with it a simple such path, whose persons are distinct and each of which holds an $f$-denominated coin of the next.  By Proposition~\ref{prop:chain-redemption}, after a sequence of redemptions along it $q$ holds an $f$-denominated coin of the sovereign.
\end{proof}

\mypara{The two legs of the peg}
The \emph{direct leg} is the issuer's obligation (Definition~\ref{def:fiat-bonds}): an $f$-denominated grassroots coin is a claim of one unit of $f$ on its issuer, worth the issuer's credit, fulfilled in kind or by a fiat coin.  The \emph{mechanised leg} is Proposition~\ref{prop:peg}: for a sovereign-connected issuer, redemption within the currency realises the claim in the sovereign's base money, whatever fiat the issuer holds.  A person's coins trade at par with the unit of $f$ while they meet redemptions as these arrive, and their capacity to meet them is measured by the ratios of Section~\ref{sec:runs}.  Proposition~\ref{prop:peg} moves one coin and consumes the coins along its chain, so being sovereign-connected secures the price of a coin and not the volume redeemable at once; what a group can obtain now is the maximum flow of Section~\ref{sec:runs}.

\mypara{The price of a coin}
Arbitrage makes par the price of the coins of an issuer who redeems them on demand, which we call meeting presentations, whether by fiat coins or by coins of issuers who meet theirs.  The persons who meet presentations in $f$ are defined in stages: those who fulfil every presentation by a fiat coin, then those who fulfil every presentation by a fiat coin or by a coin of a person of an earlier stage.

\begin{definition}[Meeting Presentations]\label{def:meets}
Let $M^0_f$ be the set of persons who fulfil every presentation of their $f$-denominated coins by a fiat coin, and $M^{k+1}_f$ the set of persons who fulfil every such presentation by a fiat coin or by an $f$-denominated coin of a person in $M^k_f$.  A person \temph{meets presentations in $f$} if they belong to $M_f = \bigcup_k M^k_f$.
\end{definition}

The argument that follows is the model's, and frictionless: a swap or a redemption costs nothing, the obligations of Definition~\ref{def:fiat-bonds} are met as assumed, newly minted coins can be sold at the quoted price while that price stands, and the arbitrage itself does not exhaust the issuer's capacity to meet presentations.

\begin{proposition}[Par by Arbitrage]\label{prop:par}
Let $p$ meet presentations in $f$.  Then the arbitrage-free price of an $f$-denominated $p$-coin is one unit of $f$.
\end{proposition}
\begin{proof}
Suppose an $f$-denominated $p$-coin trades at $\pi < 1$ units of $f$.  A buyer acquires one and presents it to $p$, receiving a fiat coin or a coin of a person of an earlier stage, which they present in turn; with $p \in M^k_f$, at most $k+1$ presentations end in a fiat coin, one unit of $f$ for $\pi$: a gain of $1 - \pi$ per coin, without risk, repeated while the price stands.  Suppose it trades at $\pi > 1$.  Minting costs $p$ nothing and obliges them for one unit of $f$ per coin (Definition~\ref{def:fiat-bonds}), so $p$ mints and sells, gaining $\pi - 1$ per coin, likewise while the price stands.  Only $\pi = 1$ admits neither trade.
\end{proof}

\iffull The sovereign can always issue a fiat coin, so it is in $M^0_f$, and a chain of mutual liquidity to it carries par along the chain with no fiat coin paid but the sovereign's.

\begin{corollary}[Par by the Sovereign]\label{cor:par-sovereign}
Let $p$ be sovereign-connected in $f$, and let each of $p$ and the persons on a simple path from $p$ to $f_\sigma$ in the $f$-coin holding graph fulfil every presentation of their $f$-denominated coins by an $f$-denominated coin of the next person on the path, or by a fiat coin.  Then the arbitrage-free price of an $f$-denominated $p$-coin is one unit of $f$, and only the sovereign pays a fiat coin.
\end{corollary}
\begin{proof}
$f_\sigma \in M^0_f$, and a person on the path who fulfils every presentation by a coin of the next person is in $M^{k+1}_f$ when the next is in $M^k_f$; by induction from the sovereign backwards along the path, $p \in M_f$, and Proposition~\ref{prop:par} applies.  Every presentation but the last is fulfilled by a coin, and the last by the sovereign's fiat coin.
\end{proof}
\fi

So, while redemptions are met, the currencies of the persons and of the sovereign are one currency, with the fiat currency as its unit of account.  The coins of a person who leaves presentations unfulfilled are still debts of $f$, and they trade on that person's credit, as a failing bank's deposits do.  Pegging each of two denominations at par imposes no constraint on the exchange rate between the two fiat currencies: by Corollary~\ref{cor:no-cross-arbitrage} instruments of two denominations never meet in a transaction, so no arbitrage links the two pegs, and each sovereign pegs its own denomination independently (Appendix~\ref{app:fx}).

\iffull\mypara{The hierarchy of money}
A denomination whose sovereign has issued its grassroots coins reproduces the hierarchy of money in grassroots form: the sovereign grassroots coins are the base money of the denomination, redeemable for the fiat currency; the coins of sovereign-connected persons are inside money, in the relation commercial bank deposits bear to cash~\cite{mcleay2014money}, without banks.  A commercial bank pegs its deposits to cash by its obligation to redeem at par on demand, backed by reserves; a grassroots issuer bears the same obligation for their coins, fulfilled in kind or by a fiat coin, and, when sovereign-connected, realised mechanically in base money.  Par follows from denomination and access to the sovereign's base money while redemption obligations are met, with no separate stabilisation mechanism.
\fi

\subsection{Liquidity, Solvency and Redemption Runs}\label{sec:runs}

\iffull
The measures are per person and the results per community; both are established for grassroots currencies and bonds in~\cite{shapiro2024gc,shapiro2026bonds}, and here every coin and bond is $f$-denominated, the graph is the $f$-holding graph, and the sink is the sovereign $f_\sigma$ (Definition~\ref{def:holding-graph}).

\mypara{The ratios}
The cash, quick and current liquidity ratios of corporate finance apply to grassroots coins and bonds, with maturity in the role of asset class, and the creditworthiness measures of international monetary economics --- foreign debt, current account and velocity --- apply to grassroots currencies~\cite{shapiro2024gc,shapiro2026bonds}.  Under denomination they are stated in units of $f$, at face value, comparable across persons and with conventional balance sheets, and for a sovereign-connected issuer who meets redemptions the face value is the price (Corollary~\ref{cor:par-sovereign}).  Ratios are computed per denomination: coins and bonds of distinct denominations never meet in a transaction (Corollary~\ref{cor:no-cross-arbitrage}), so one ratio over several denominations would count assets that cannot meet the liabilities in its denominator.  The ratios separate insolvency, which no loan remedies, from illiquidity, which a loan does~\cite{shapiro2026bonds}.

\mypara{Solvency}
A \emph{computation} is a sequence of configurations, each reached from the previous by a transaction of Definition~\ref{def:sgc-transactions}; it is \emph{safe} if every step is one.  The results are stated against the sink, the member with whom a community's net position is left.  A set of persons containing $f_\sigma$ is \emph{sink-closed} if every edge of the $f$-holding graph between two vertices other than $f_\sigma$ has both endpoints in the set or neither: a community whose only claims across its boundary, in $f$, are with the sovereign.  A redemption in which the debtor is not $f_\sigma$ and the instrument taken is not the debtor's own is a \emph{clearing redemption}, and a bond returns to its issuer only as a coin, by Mature and then by a clearing redemption.  The holdings of a sink-closed set $V$ \emph{clear} if some finite safe computation of clearing redemptions, Mature transactions and advance-date transitions of members of $V$ leads to a configuration in which every $f$-denominated coin and bond of an issuer in $V$ other than $f_\sigma$ is held by its issuer~\cite{shapiro2024gc,shapiro2026bonds}.

\begin{proposition}[Graph Solvency~\cite{shapiro2024gc,shapiro2026bonds}]\label{prop:graph-solvency}
The holdings of a sink-closed set $V$ clear if and only if $a_p \ge \ell_p$ for every member $p$ of $V$ other than $f_\sigma$, and in cleared holdings each such member holds exactly $a_p - \ell_p$ sovereign grassroots coins and bonds.
\end{proposition}

Clearing is by redemptions, in any order and with no coordination: when the condition holds, every computation in which the dates of the members advance beyond any bound and a member who can redeem or take a matured bond for a coin eventually does reaches cleared holdings, after at most $\sum_{p \in V \setminus \{f_\sigma\}} \ell_p$ redemptions and as many Mature transactions~\cite{shapiro2024gc,shapiro2026bonds}.  The propositions quantify over safe computations and do not assume a holder can inspect the debtor's holdings; in practice a debtor incapable of set-off discloses their holdings to the redeemer.

\mypara{Liquidity}
For clearing, any coin or bond of the debtor may be taken; here a group needs sovereign grassroots coins, now, with no date advancing.  Given a set $A$ of persons not containing $f_\sigma$, the \emph{liquidity of $A$} is the maximum, over finite safe computations of redemptions of members of $A$, of the number of sovereign grassroots coins members of $A$ hold in the last configuration~\cite{shapiro2024gc}.

\begin{proposition}[Graph Liquidity~\cite{shapiro2024gc,shapiro2026bonds}]\label{prop:graph-liquidity}
The liquidity of $A$ is the maximum flow from $A$ to $f_\sigma$ in the $f$-coin holding graph.
\end{proposition}

The edges are coins already outstanding, not capacities to borrow or to route, and a member realises a path by exercising their own redemption rights one after another, with no intermediary asked to consent; a redemption that takes a bond obtains no coin and consumes the coin surrendered, so the flow is over the coins~\cite{shapiro2026bonds}, and not over the matured bonds a holder could first take for coins: whether a bond is mature is subject to its issuer advancing its date, and the liquidity of a group does not depend on that.  For a single person the maximum flow is the number of sovereign grassroots coins they can reach by themselves, and Proposition~\ref{prop:peg} is the case of one coin.

\mypara{Redemption runs}
A \emph{redemption run on} $p$ is a maximal finite safe computation of redemptions, each surrendering an $f$-denominated $p$-coin to $p$ and taking a coin of an issuer other than $p$, within a day in which no bond is taken for a coin; it is \emph{met in full} if no $f$-denominated $p$-coin is in circulation in its last configuration~\cite{shapiro2024gc,shapiro2026bonds}.  Let $\hat a_p$ and $\hat\ell_p$ be the out-degree and in-degree of $p$ in the $f$-coin holding graph: the coins $p$ holds and the coins $p$ owes, their bonds being neither an asset nor a liability in the run.

\begin{proposition}[Redemption Run~\cite{shapiro2024gc,shapiro2026bonds}]\label{prop:run}
Every redemption run on $p$ is met in full if $\hat a_p \ge \hat\ell_p$, and otherwise ends with exactly $\hat\ell_p - \hat a_p$ $f$-denominated $p$-coins in circulation, their presentations unfulfilled.
\end{proposition}

A shortfall stands as unfulfilled presentations, evidencing the issuer's default in $f$.  A bearer who instead takes a bond the issuer holds is met by it, as obligation~2 allows, and the run is shorter by that presentation.  A redemption run therefore stops at the issuer's holdings, and reaches the central bank only through the sovereign grassroots coins the issuer holds.
\fi

\subsection{Among the Contenders}\label{sec:contenders}

\iffull
Tables~\ref{tab:contenders} and~\ref{tab:contenders-econ} compare the four retail architectures~\cite{auer2020technology,auer2020rise}, the private alternatives, and the two kinds of coin of the proposed architecture, on the structure of a payment and on credit, cost, throughput, runs, disintermediation, the legal instrument and recovery.  Bonds take no part in a retail payment and are out of both.
\fi

\iffull Every design compared but the last two places one record and one party inside every retail payment: the central bank, an intermediary, an issuer or a chain.  Here every person is an issuer and keeps their own record, so the party inside a payment is the issuer of the coins paid --- the central bank for a payment in sovereign grassroots coins, as in a direct central bank digital currency, and a community bank, a merchant or a neighbour for a payment in theirs.  The central bank is thereby out of retail operations without an intermediary tier, and credit creation, which every other design leaves to the banks outside it, is a transaction of the currency.

Designs that preserve privacy make payments unlinkable to the operator and regulate them without revealing them: Platypus~\cite{wust2022platypus} combines the transaction processing of electronic cash with an account-based fund model, and PEReDi~\cite{kiayias2022peredi} distributes the record among authorities holding shares of it.  Here the issuer of the coins paid is the one party to the payment beside the payer and the payee, and what the central bank sees is the circulation of its own coins.
\fi

\iffull\begin{sidewaystable}[p]
\scriptsize
\setlength{\tabcolsep}{3pt}
\renewcommand{\arraystretch}{1.05}
\newcolumntype{N}[1]{>{\raggedright\arraybackslash}p{#1}}
\renewcommand{\thetable}{\arabic{table}A}
\begin{tabular}{N{1.6cm}N{1.9cm}N{2.4cm}N{2.1cm}N{1.9cm}N{2.0cm}N{2.5cm}N{2.8cm}}
\hline
 & \textbf{A.\ Direct} & \textbf{B.\ Hybrid} & \textbf{C.\ Intermediated} & \textbf{D.\ Indirect} & \textbf{E.\ Private DC} & \textbf{F.\ Sovereign grassroots coins} & \textbf{G.\ Non-sovereign grassroots coins} \\
\hline
\textbf{1.\ Retail holder's debtor} & the central bank & the central bank & the central bank & an intermediary, backed by central bank money & the issuer, holding reserves (stablecoin); a bank (tokenised deposit) & the central bank & the person who issued the coins held, natural or legal \\
\textbf{2.\ Retail record kept by} & the central bank & intermediaries, with a copy at the central bank; the digital euro's ledger is the Eurosystem's own, centralised across three regions~\cite{ecb2025closing}, and the digital pound's is the central bank's~\cite{boe2023digitalpound} & intermediaries; the central bank keeps a wholesale ledger & the intermediary & the issuer's ledger or a blockchain; the banks & each holder, their own holdings; the central bank the record of its coins & each holder, their own holdings; the issuer the record of their coins \\
\textbf{3.\ Retail payment processed by} & the central bank & an intermediary; the digital euro and the digital pound settle at the central bank~\cite{ecb2023stocktake,boe2023digitalpound} & an intermediary & the intermediary & the ledger's validators; the banks & the central bank & the issuer of the coins paid \\
\textbf{4.\ Participants in a payment} & payer, payee, the central bank & payer, payee, their intermediaries & payer, payee, their intermediaries & payer, payee, their intermediaries & payer, payee, the ledger or the banks & payer, payee, the central bank and the machines holding its log & payer, payee, the issuer and the machines holding their log: a person's state custodians~\cite{eitan2026securing}, an institution's own \\
\textbf{5.\ Offline payment} & by value stored on the device, at the holder's or the operator's risk & the same & the same & the same & no & no: the central bank and the machines holding its log must be reachable & not offline from the issuer; over Bluetooth or a local network that reaches the issuer and the machines holding their log, including one with no connection to the Internet \\
\textbf{6.\ Who sees a payment} & the central bank & the intermediary and the central bank & the intermediary; the central bank in aggregate & the intermediary & the ledger's operator, or everyone on a public chain; the banks & payer, payee and the central bank & the three parties; the machines holding the issuer's log record it, and their persons cannot read it~\cite{eitan2026securing} \\
\textbf{7.\ Retail balances by identity at the central bank} & not necessarily; the central bank keeps the retail record & no & no & no & no & no: the central bank keeps the record of its own coins & no \\
\hline
\end{tabular}
\caption{CBDC architectures~\cite{auer2020technology,auer2020rise}: part A, the structure of a payment.  Columns F and G are the two kinds of coin of the proposed architecture, the central bank's and every other person's; a grassroots currency comprises the coins and the bonds of its issuer, and bonds take no part in a retail payment.  Of the CBDCs compared, the e-CNY~\cite{pboc2021ecny}, the Sand Dollar~\cite{sanddollar}, the eNaira~\cite{cbn2021enaira} and JAM-DEX~\cite{boj2022jamdex} are hybrid~\cite{bis2024cgide}, and so is Chaumian eCash~\cite{chaum2021cbdc,bis2023tourbillon}, with an anonymous payer; the digital euro~\cite{ecb2023stocktake} and the digital pound~\cite{boe2023digitalpound} are two-tier, distributed by private intermediaries with the ledger kept and payments settled at the central bank; Project Hamilton~\cite{lovejoy2022hamilton} is a system design for a direct core; the intermediated~\cite{auer2021quest} and the indirect~\cite{adrian2019rise,auer2020technology} architectures are proposals; the private alternatives are regulated stablecoins~\cite{genius2025act,mica2023,gorton2023taming} and tokenised deposits~\cite{garratt2023singleness}.  Rows 5 and 6 are properties of the implementations cited and not of the architecture class, offline support and who sees a payment being design choices within each.}
\label{tab:contenders}
\end{sidewaystable}

\begin{sidewaystable}[p]
\scriptsize
\setlength{\tabcolsep}{3pt}
\renewcommand{\arraystretch}{1.05}
\newcolumntype{N}[1]{>{\raggedright\arraybackslash}p{#1}}
\addtocounter{table}{-1}\renewcommand{\thetable}{\arabic{table}B}
\begin{tabular}{N{1.6cm}N{1.9cm}N{2.4cm}N{2.1cm}N{1.9cm}N{2.0cm}N{2.5cm}N{2.8cm}}
\hline
 & \textbf{A.\ Direct} & \textbf{B.\ Hybrid} & \textbf{C.\ Intermediated} & \textbf{D.\ Indirect} & \textbf{E.\ Private DC} & \textbf{F.\ Sovereign grassroots coins} & \textbf{G.\ Non-sovereign grassroots coins} \\
\hline
\textbf{8.\ Credit creation within the architecture by} & no one: a central bank liability & no one & no one & no one: fully backed & no one (stablecoin); the banks (tokenised deposit) & the central bank, by its mutual credit lines & every person, by issuing their own coins \\
\textbf{9.\ Cost of a payment} & the central bank's processing & the intermediary's, fees by policy & the same & the same & the chain's fees; the banks' fees & operating a device; no fee required by the architecture & operating a device; no fee required by the architecture \\
\textbf{10.\ Throughput bounded by} & the central bank's processor & the central bank's ledger and the intermediaries & the intermediaries & the intermediaries & the chain; the banks' systems & the central bank and the machines holding its log & no component shared by all issuers; grows with the issuers and their custodians \\
\textbf{11.\ Where a run goes} & from bank deposits into the CBDC, at the speed of conversion & the same & the same & into the intermediary's backing & onto the issuer's reserves; onto the bank & \multicolumn{2}{N{5.5cm}}{onto one issuer, met from their holdings; what the public can obtain of central bank money is bounded by what the central bank has issued and not taken back} \\
\textbf{12.\ Disintermediation mitigant} & holding limits, tiered remuneration, waterfall & the same & the same & none needed: the claim is on the intermediary & a tokenised deposit needs none, the claim staying on the bank~\cite{garratt2023singleness}; for a stablecoin none is proposed, the mitigants proposed addressing runs---issuance through banks, or full backing with safe assets~\cite{gorton2023taming} & \multicolumn{2}{N{5.5cm}}{none needed to bound demand-driven creation: no holder of a bank deposit has a right to newly issued sovereign grassroots coins, which enter in exchange for fiat coins, by the central bank's payments and by its mutual credit lines, all at its own election} \\
\textbf{13.\ Legal instrument} & a new central bank liability, by legislation & the same & the same & e-money or deposit law & stablecoin regulation; deposit law & a signed debt of the central bank, analogous to a promissory note; the authority to issue it is statutory~\cite{boe2026progress} & a signed debt of its issuer, analogous to a promissory note, under the ordinary rules of contract \\
\textbf{14.\ Recovery of a lost key or device} & by identity, from the central bank's record, where access is account-based; as the token design provides, where token-based~\cite{auer2020technology,bis2024cgide} & by identity, from the intermediary's record or the central bank's copy, where account-based (the eNaira, the e-CNY); as the token design provides, where token-based (the Sand Dollar, JAM-DEX)~\cite{bis2024cgide} & by identity, from the intermediary's record, where account-based; as the token design provides, where token-based & the same & a bearer instrument~\cite{garratt2023singleness}, held by whoever holds the key; a tokenised deposit from the bank's record & the holder's coins from the central bank's log, their identity from their identity custodians~\cite{eitan2026securing} & holdings from the issuers' logs; a person's own log, exactly, and their identity from their state and identity custodians~\cite{eitan2026securing}; an institution's log from its own machines \\
\hline
\end{tabular}
\caption{CBDC architectures, part B: credit, cost, throughput, runs, disintermediation, the legal instrument, and recovery from the loss of a key or device; the columns are those of part A.  Rows 11 and 12 hold of the two kinds of coin together and are not properties of either.  Rows 9 and 14 are properties of the implementations cited and not of the architecture class, the cost of a payment and the recovery of a lost key being design choices within each.}
\label{tab:contenders-econ}
\end{sidewaystable}
\renewcommand{\thetable}{\arabic{table}}
\fi

\section{Interest: Grassroots Bonds}\label{sec:interest}

With grassroots bonds, credit acquires time and a price.  Here we add bonds to the coins of Sections~\ref{sec:money} and~\ref{sec:credit}, present the instruments formed from coins and bonds, and the central bank's instruments of monetary policy as transactions in its coins and bonds.

\subsection{Maturity and Interest}\label{sec:maturity}

A coin and a bond are distinct instruments (Section~\ref{sec:denominated}): a coin is a unit of its issuer's debt due now and carries no date, and a bond is that unit due at its date.  A coin may be paid, and presented to its issuer for redemption.  A bond may be minted, swapped, escrowed and sold, and an issuer may fulfil a presentation with one; it is neither paid nor presented, and once it is mature its bearer may take a coin of its issuer for it~\cite{shapiro2026bonds}.

Coins are demand instruments and bear no interest.  Interest is borne by bonds, as the difference between the face value of a bond and the coins swapped for it: a lender who advances $k$ coins against bonds of face value $k(1+\rho)$ maturing at $d$ lends at the rate $\rho$ for that maturity; $\rho$ is the return over the term, not an annual rate, and rates of unlike terms are compared after annualisation.  A swap of coins for coins is at par and on demand, so it carries no interest --- prompt redemption after it would return any premium (Section~\ref{sec:credit-lines})~\cite{shapiro2026bonds} --- and the discount at which coins are swapped for a bond, with its maturity, is the rate.

\full{Maturity is judged by the date of the issuer, which the issuer warrants when a matured bond of theirs is taken for one of their coins; an issuer whose date lags behind the calendar breaches that warranty, and the record of the transaction witnesses the breach against them~\cite{shapiro2026bonds}.  The obligation of a coin is enforced by the contract; that of a bond as a whole is attested, and enforceable at law on the record (Appendix~\ref{app:schemas}).}{Maturity is judged by the date of the issuer, which they warrant when a matured bond of theirs is taken for one of their coins, and the record witnesses a breach of that warranty against them~\cite{shapiro2026bonds}.}

\subsection{The Instruments}\label{sec:instruments}

\iffull
Every person can lend at interest, sell debt, open term credit lines, post collateral, issue deposit-like claims, and form the further financial instruments; each is a swap of bonds or an escrow agreement~\cite{shapiro2026bonds}.  A zero-coupon loan is a swap of freshly minted coins for the borrower's bonds of larger face value maturing later; a sale of debt before maturity is a swap of a third person's bonds for coins at a discount; a forward is a swap of bonds of the two parties maturing at one date.  With an escrow agent, trusted by the parties for that agreement and holding the bonds it transfers by the agreed conditions, come loans with payment schedules, term credit lines with draws, options, collateral, guarantees, insurance, credit default swaps and letters of credit.  Every instrument is stated in Appendix~\ref{app:escrow}, recalled from~\cite{shapiro2026bonds}.

A term credit line is the escrow instrument the central bank uses, and differs from the mutual credit line of Section~\ref{sec:credit-lines}, which is a swap of coins at par and on demand: the lender deposits its coins with the escrow agent up to the agreed limit, the borrower draws by depositing principal and coupon bonds, the agent releases the two sides crosswise, and at expiry the undrawn balance returns to the lender~\cite{shapiro2026bonds}.  A draw is therefore a transaction the lender has consented to in advance and the borrower elects, and each tranche falls due at its own maturity.

A bank's coins play the economic role of its demand deposits (Section~\ref{sec:every-person}) and its bonds that of its term deposits, and it lends by swapping its coins for its borrowers' bonds, at rates it sets.  A saver who holds a bank's bond holds a dated claim on the bank, and sells it before maturity at a discount rather than breaking a deposit.
\fi

\subsection{The Central Bank's Instruments}\label{sec:policy}

The central bank's instruments of monetary policy are the same transactions, and it chooses the counterparty, amount, rate, maturity and collateral of its credit operations; the counterparty of an open market operation is whoever holds the security.  Lending to a named city, firm or household is credit policy as much as monetary policy, and which it is, and whether the central bank may do it at all, is a matter of its mandate, which we take as given.

\iffull Goodfriend~\cite{goodfriend2011credit} draws the distinction: an open market operation in government securities returns its revenue to the fiscal authorities, while lending to particular borrowers allocates public funds and is debt-financed fiscal policy.  He would therefore confine a central bank to last-resort lending to solvent supervised institutions, unless the fiscal authorities agree in advance.  C\'urdia and Woodford~\cite{curdia2011balance} treat such lending as a dimension of policy distinct from the interest rate, worth using when private intermediation is impaired but costly in itself; the case for it depends on the market lent into, and the size of a spread does not settle it.  The central bank can make such loans under this architecture; whether it should do so is a question of its mandate.
\fi

\mypara{The policy rate}
The policy rate is the rate at which the central bank lends its coins against bonds: it advances $k$ of its coins against a counterparty's bonds of face value $k(1+\rho)$ maturing at $d$, and $\rho$, the return over that term, is the rate.  A change of policy changes the terms of new swaps, and no bond already issued.

\mypara{Targeted liquidity}
A mutual credit line is targeted liquidity without a rate (Section~\ref{sec:credit-lines}); with maturity it acquires one.  Targeted liquidity at a rate is a term credit line to a city, a bank, a firm or a community, whose tranches are draws the central bank has consented to, each falling due at its maturity.  A city struck by a disaster draws the central bank's coins against its own city bonds and funds relief in coins redeemable at face value in sovereign money.  As the city repays, it acquires sovereign grassroots coins and redeems them from the central bank against its principal bonds, which the central bank cannot refuse, and the undrawn balance returns to the central bank at expiry; a tranche already drawn falls due at its own maturity, and the central bank does not recall it.  The recipient, the amount, the maturity and the rate are the parameters of the transactions, and who drew, who repaid and what remains outstanding is their record.  Where conventional policy transmits unevenly across time and sectors, and weakly in recessions~\cite{olivei2007timing,tenreyro2016pushing}, here the central bank contracts with the intended recipient and no intermediation stands between.

\mypara{Lender of last resort}
A redemption run exhausts the coins an issuer holds, and their remaining bonds may cover their liabilities while their coins do not: the failure is of liquidity and not of solvency, and the ratios distinguish the two (Section~\ref{sec:runs}).  In a panic no one lends even against good bonds, so the central bank does: it opens an emergency term credit line with the distressed issuer, advancing its coins against their bonds, collateralised by their remaining holdings, at a penalty rate above the policy rate.  This is Bagehot's prescription~\cite{bagehot1873lombard}, to lend freely against good collateral at a penalty rate, as a swap of bonds; it differs from the lending above in the counterparty's condition, the collateral and the rate.

\iffull\mypara{The central bank's own bonds}
With its own bonds --- interest-bearing liabilities it issues to the counterparties it chooses --- the central bank absorbs liquidity from those counterparties, on terms set at issuance, with no tier and no cap imposed by the architecture.  A counterparty swaps the central bank's coins now for its bonds maturing later, and holds the safest dated claim in the denomination.

\mypara{Open market operations}
The central bank buys and sells government and bank bonds for its coins, which are its open market operations; a purchase puts its coins into circulation against a bond it then holds, and a sale takes them back.  Securities that are not grassroots bonds are bought and sold for its coins in the same way, by a payment against delivery outside the currency.
\fi

\iffull
\mypara{Where interest is borne}
\fi
Interest is charged on the borrowing of sovereign grassroots coins and paid on the bonds the central bank issues to the counterparties it chooses.  The coins themselves bear none, as cash does not, so they introduce no interest-rate competition with deposits; a claim that bears no interest can still be preferred for its safety or its convenience\full{ --- India's digital rupee, which bears none and is capped, reduced the retail deposits of the banks exposed to it by 2.7\%, concentrated in liquid savings accounts~\cite{dimaggio2026rupee} --- }{, }and the central bank's issuance bounds the coins in circulation (Section~\ref{sec:credit-lines}).  A remunerated central bank digital currency pays its rate to whoever holds it and so competes with deposits, which is one reason neither the digital euro nor the digital pound is to be remunerated~\cite{ecb2025closing,boe2023digitalpound}; the Eurosystem states that holdings would not be remunerated and would be subject to holding limits, to preserve the role of banks in providing credit~\cite{ecb2025closing}.

\iffull
\mypara{What the other proposals leave outside}
\fi
The retail proposals compared add a retail form of central bank money and leave the central bank's lending, its asset purchases and the implementation of the policy rate in reserves, an instrument the public does not hold.  Here they are transactions in the coins and bonds the public holds.

\subsection{A Counterparty-Specific Rate Corridor}\label{sec:corridor}

The central bank sets two rates for a counterparty: the rate at which it lends its coins against that counterparty's bonds, and the rate it pays on its own bonds issued to them.  We argue that the two bound the rates the counterparty faces, at like maturity, collateral and risk.  The corridor is counterparty-specific.  We argue a bound on the rates one admitted counterparty faces; we do not argue the overnight interbank corridor of today's operational frameworks~\cite{abad2025cbdc}, nor the transmission of policy to output and prices.  The argument assumes that a person takes the cheaper of two credits available to them on like terms, and the dearer of two claims.

\iffull\mypara{The upper bound}
A counterparty who can still borrow from the central bank does not pay more than its lending rate.  A lender who asks more offers, at like maturity, collateral and risk, a credit the counterparty already has on better terms, and the counterparty draws on the central bank instead.  So the rate the counterparty pays is at most the central bank's lending rate.

\mypara{The lower bound}
A counterparty who can still acquire the central bank's bonds does not accept less than the rate on them.  A borrower who offers less offers, at like maturity, a dated claim no safer than the central bank's --- whose issuer can always issue a fiat coin (Section~\ref{sec:the-sovereign}) --- at a lower rate, and the counterparty buys the central bank's bond instead.  So the rate the counterparty accepts is at least the rate on the central bank's own bonds.
\fi

\iffull\mypara{Whom the corridor binds}
Both bounds hold while the counterparty can act on them: the term credit line is not exhausted and its collateral suffices, and the central bank is still issuing its bonds to them.  Both conditions are quantities the central bank sets, so it admits to the corridor the counterparties it chooses, and a counterparty it does not admit borrows and lends on their own credit.  Today the standing facilities~\cite{abad2025cbdc} bound the rate at which banks lend to one another, one corridor for the interbank tier; here the corridor is counterparty-specific and the counterparties are whoever the central bank admits --- a bank, a city, a firm, a cooperative, a household --- and the instruments are the coins and bonds the public holds, so the corridor binds at the tier the central bank chooses, not solely at the interbank tier.
\fi

\subsection{Among the Types of Money}\label{sec:money-types}

\iffull
Table~\ref{tab:money} sets the three components beside the types of money proposed as a retail central bank digital currency or in its place --- account-based and token-based designs, regulated stablecoins, and tokenised deposits --- by whose liability each is, what backs it, where its record is kept, how it is created, whether it bears interest, and what can be done with it.\full{  Cash, deposits, cryptocurrencies and community currencies are not proposed as a central bank digital currency and are out of it.}{}  A grassroots bond is a dated claim and is not paid; its bearer takes a coin of its issuer for it once it is mature, so its column is credit and not money.
\fi

\iffull\begin{sidewaystable}[p]
\scriptsize
\setlength{\tabcolsep}{3pt}
\renewcommand{\arraystretch}{1.05}
\newcolumntype{N}[1]{>{\raggedright\arraybackslash}p{#1}}
\begin{tabular}{N{1.5cm}N{2.3cm}N{2.5cm}N{1.9cm}N{1.8cm}N{2.2cm}N{2.4cm}N{2.4cm}}
\hline
 & \textbf{A.\ Account-based} & \textbf{B.\ Token-based} & \textbf{C.\ Regulated stablecoins} & \textbf{D.\ Tokenised deposits} & \textbf{E.\ Sovereign grassroots coins} & \textbf{F.\ Non-sovereign grassroots coins} & \textbf{G.\ Grassroots bonds} \\
\hline
\textbf{1.\ Whose liability} & the central bank's~\cite{bis2020cbdc} & the central bank's~\cite{bis2020cbdc} & the issuer's~\cite{gorton2023taming} & the bank's~\cite{garratt2023singleness} & the central bank's~\cite{shapiro2024gc} & each person's, natural or legal~\cite{shapiro2024gc} & its issuer's, the central bank or any other person~\cite{shapiro2026bonds} \\
\textbf{2.\ Backing and redemption} & central bank money itself; convertible at par with cash~\cite{bis2020cbdc} & central bank money itself; convertible at par with cash~\cite{bis2020cbdc} & reserves of cash and short-term government debt; redeemable at par~\cite{genius2025act,mica2023}, though not always accepted at par in the market~\cite{gorton2023taming} & as deposits; transfers settled in central bank money~\cite{garratt2023singleness} & the fiat currency the central bank issues; redeemed at par into fiat coins on presentation & the issuer's goods and services and the coins they hold; redeemed at par against any coin they hold, and along chains of liquidity into sovereign grassroots coins~\cite{shapiro2024gc} & as its issuer's coins: once the date of its issuer has reached its maturity, its bearer may take a coin of the issuer for it~\cite{shapiro2026bonds} \\
\textbf{3.\ Record} & balances by identity, on the central bank's ledger or intermediaries'~\cite{auer2020technology,bis2024cgide}: the eNaira, the e-CNY, the digital euro and the digital pound~\cite{bis2024cgide,ecb2025closing,boe2026progress} & tokens held by key~\cite{auer2020technology,bis2024cgide}: blind-signed coins with a spent list at the central bank~\cite{chaum2021cbdc,bis2023tourbillon}, an unspent-output set~\cite{lovejoy2022hamilton}; the Sand Dollar and JAM-DEX~\cite{bis2024cgide} & a blockchain; a bearer instrument~\cite{garratt2023singleness} & the bank's ledger, or a unified ledger~\cite{garratt2023singleness} & the central bank's log, on its own machines & each person's log of their own coins, backed up by their custodians~\cite{shapiro2024gc,eitan2026securing}; an institution's log on its own machines & the log of its issuer, as its coins \\
\textbf{4.\ Creation and entry into circulation} & in the implementations compared, issued by the central bank against bank money converted into it~\cite{bis2024cgide} & in the implementations compared, issued by the central bank against bank money converted into it~\cite{bis2024cgide} & issued against fiat received & created by the bank's lending, as deposits~\cite{mcleay2014money} & minted by the central bank, entering in exchange for fiat coins, by its payments, and by its mutual credit lines against the counterparty's coins or, with bonds, their bonds & minted by each person, entering by mutual credit, coins for coins & minted by its issuer, entering by loans, coins for bonds \\
\textbf{5.\ Interest} & a design choice; the implementations compared bear none~\cite{pboc2021ecny,boe2023digitalpound} & a design choice; the implementations compared bear none~\cite{boj2022jamdex,chaum2021cbdc} & the token bears none & as deposits & none, as on cash & none & the discount at which coins are swapped for them~\cite{shapiro2026bonds} \\
\textbf{6.\ Operations} & payment and conversion, under holding limits, tiers and waterfalls~\cite{bis2021finstab} & payment and conversion & payment on the chain; redemption by the issuer's direct customers~\cite{gorton2023taming} & payment and transfer; the bank's lending beside it & payment and redemption; the central bank's mutual credit lines & payment, mutual credit and redemption & credit: for the central bank the policy rate, targeted liquidity, the lender of last resort, its own bonds and open market operations; for every other person loans at interest, sale of debt, term credit lines, collateral and the escrow instruments~\cite{shapiro2026bonds} \\
\hline
\end{tabular}
\caption{Types of money proposed as a retail central bank digital currency or in its place, and the credit instrument beside them: whose liability each is, what backs it and how it is redeemed, where its record is kept, how it is created and enters circulation, whether it bears interest, and what can be done with it; all have the fiat currency as their unit of account.  Columns A and B are the designs of Tables~\ref{tab:contenders} and~\ref{tab:contenders-econ} by the form of the money.  Columns E, F and G are the three components of the proposed architecture; a grassroots bond is a dated claim: it is not paid, and its bearer takes a coin of its issuer for it once the date of its issuer has reached its maturity, so column G is credit and not money.}
\label{tab:money}
\end{sidewaystable}
\fi

\section{Related Work}\label{sec:related-work}

\iffull\mypara{Central bank digital currency}
Auer, Frost, Gambacorta, Monnet, Rice and Shin~\cite{auer2022motives} survey the motives for a central bank digital currency and its economic implications.  The requirements of a retail central bank digital currency~\cite{bis2020cbdc}, its financial-stability implications~\cite{bis2021finstab} and its system design~\cite{bis2021design} are the reports this paper is checked against in Section~\ref{sec:cbdc}; the four architectures~\cite{auer2020technology,auer2020rise,auer2021quest,bis2024cgide}, the currencies in operation~\cite{sanddollar,cbn2021enaira,boj2022jamdex,pboc2021ecny} and in preparation~\cite{ecb2023stocktake,ecb2025closing,boe2023digitalpound,boe2026progress}, the system designs~\cite{lovejoy2022hamilton,chaum2021cbdc,bis2023tourbillon,wust2022platypus,kiayias2022peredi}, and the private alternatives are compared with the proposal in Section~\ref{sec:contenders} and in Tables~\ref{tab:contenders} and~\ref{tab:contenders-econ}, and surveyed in~\cite{senn2026sok}.  Kumhof and Noone~\cite{kumhof2018cbdc} are the precedent for issuance that is not by conversion: their four principles include no guaranteed on-demand convertibility of bank deposits into the currency, and issuance only against eligible securities.  Their central bank buys securities its counterparty already holds; here it lends against coins its counterparty issues, and the counterparty is whoever it chooses (Section~\ref{sec:credit-lines}).

\mypara{The currency as an instrument of monetary policy}
A central bank digital currency that bears interest has been proposed as an instrument of policy.  Bordo and Levin~\cite{bordo2017cbdc} would make it account-based and interest-bearing, make its rate the main tool of policy, and remove the effective lower bound by a graduated schedule of fees on transfers between cash and the currency.  Davoodalhosseini~\cite{davoodalhosseini2022cbdc} studies optimal monetary policy where agents hold cash, such a currency, or both, and finds that where the cost of using it is not too high it implements more efficient allocations than cash and can attain the first best, and the central bank can make what it pays contingent on the balance held; Agur, Ari and Dell'Ariccia~\cite{agur2022designing} derive the optimal design where agents sort into cash, the currency and deposits by their preferences over anonymity and security, and find that a deposit-like currency depresses bank credit and output while a cash-like one may drive cash out, and an interest-bearing currency eases the tradeoff where network effects matter.  The second limitation of Section~\ref{sec:introduction} is therefore of the designs in operation and in preparation, in which the currency is a payment liability and policy operates through reserves, and not of the concept.  The coins here bear no interest.  The architecture adds the central bank's lending, its bond issuance and its open market operations, as transactions in the coins and bonds the public holds (Section~\ref{sec:policy}).

\mypara{Deposits, recycling and disintermediation}
Whether the substitution of public for private money matters depends on what the central bank does with the funding that reaches it.  Brunnermeier and Niepelt~\cite{brunnermeier2019equivalence} give conditions under which a swap of private for public money leaves the allocation and the price system unchanged, when the central bank passes the funds back to the banks that lost them, so that a central bank digital currency need imply neither a credit crunch nor a loss of financial stability; Kim and Kwon~\cite{kim2023cbdc} show that where the deposits taken are not lent back, credit supply falls, and that lending them back to the banks raises both credit supply and stability; Fern\'andez-Villaverde, Sanches, Schilling and Uhlig~\cite{fernandezvillaverde2021central} have the central bank intermediate through investment banks: its contract with them is not callable, so it cannot be run on, and it becomes the deposit monopolist, which may endanger maturity transformation; and the review of Infante, Kim, Orlik, Silva and Tetlow~\cite{infante2024retail} finds that the effect on banks depends on how the central bank recycles the new liability, by asset purchases or by lending.  Keister and Sanches~\cite{keister2023should} set the efficiency of exchange against the crowding out of deposits, and show that a currency targeted to compete only with cash or only with deposits can raise welfare where a single universal one need not; Chiu, Davoodalhosseini, Jiang and Zhu~\cite{chiu2023bank} show that where banks have market power in the deposit market a currency remunerated in an intermediate range raises the deposit rate and with it bank lending and output, and its rate acts as a floor under the deposit rate.  Bidder, Jackson and Rottner~\cite{bidder2025fastslow} separate the two horizons: slow disintermediation, which shrinks a fragile banking system in normal times and stabilises it, and fast disintermediation in a run, which a currency invites because, unlike cash, it can be held at scale at no cost, so a holding limit keeps the first without the second.  The one live retail currency studied bears this out at the slow horizon: Di Maggio, Ghosh, Ghosh, Vats and Wu~\cite{dimaggio2026rupee} find that eligibility for India's digital rupee, which bears no interest and is capped, lowered retail deposits by 2.7\%, concentrated in liquid savings accounts and not in term deposits, that credit fell and then recovered as banks changed their funding mix, and that the currency displaced deposit-backed digital payments rather than cash.  Each of these concerns a currency the public obtains by converting deposits into it.  Here there is no such flow to recycle: no holder of a deposit has a right to newly issued sovereign coins, the central bank issues them to the counterparties it elects against coins those counterparties issue (Section~\ref{sec:credit-lines}), and the credit a bank loses is created by the persons themselves.  The coins bear no interest, so they introduce no interest-rate competition with deposits; a claim that bears none may still be preferred for its safety or its convenience, and the central bank's issuance bounds the quantity in circulation.

\mypara{The operational framework of policy}
Abad, Nu\~no and Thomas~\cite{abad2025cbdc} analyse what a central bank digital currency does to the framework in which policy is implemented: its adoption contracts deposits, and the contraction is absorbed first by a fall in excess reserves and, if it is large enough, by recourse to central bank credit; the system moves from a floor to a corridor and then to a ceiling, and the effect on credit, investment and output differs with the framework it reaches.  The corridor of Section~\ref{sec:corridor} is counterparty-specific, and bounds the rates one admitted counterparty faces at like maturity, collateral and risk; theirs is the corridor between the standing facilities, which bounds the rate at which banks lend to one another.  Lending to a named city, firm or household is credit policy as much as monetary policy, and whether a central bank may allocate credit in that way is a question of its mandate, which this paper takes as given (Section~\ref{sec:policy}).  Goodfriend~\cite{goodfriend2011credit} and C\'urdia and Woodford~\cite{curdia2011balance} distinguish credit policy from monetary policy and say what each commits of public funds.

\mypara{Synthetic central bank digital currency}
A synthetic central bank digital currency~\cite{adrian2019rise} has private firms issue digital liabilities backed one for one by central bank reserves, so that private provision rests on a public anchor.  Denominated grassroots coins take the opposite choice.  Their issuers create credit, bear their own default, and hold no reserve against their coins, and their holders reach base money by redemption along chains of liquidity (Proposition~\ref{prop:peg}) rather than by backing.  The two put different things behind a privately issued claim of a fiat unit: reserves held against it, or an obligation realised by redemption.

\mypara{What the mitigants bound}
The measures proposed against deposit flight --- limits on holdings and on transactions, tiered remuneration and waterfalls~\cite{bis2021finstab,ecb2025closing}, Bindseil's two-tier remuneration, which he offers as the simpler alternative to Kumhof and Noone's principles, where the first tier stays attractive and never negative and the second is priced to discourage the currency as a store of value~\cite{bindseil2020tiered}, and Kumhof and Noone's own refusal of guaranteed convertibility between the currency and reserves and between deposits and the currency~\cite{kumhof2018cbdc} --- bound a demand that the design admits.  The Bank of England has measured the effect of such a limit: under its severe illustrative stress, with no individual limit some 21\% of the banks modelled fall below a 100\% liquidity coverage ratio and the demand for central bank lending reaches about \pounds 252 billion, against 9 to 12\% and \pounds 21 to \pounds 112 billion at individual limits of \pounds 5{,}000 to \pounds 20{,}000, and the limit is traded off against what households can do with the money~\cite{boe2025limits}.  Here there is no such demand to bound: the outstanding coins and bonds are those the central bank has issued and has not taken back (Lemma~\ref{lem:conservation}), so the quantity is settled at issuance and maturity only changes its composition; a limit for other reasons is left to the deployment.

\mypara{Grassroots currencies and bonds}
Grassroots coins and bonds~\cite{shapiro2024gc,shapiro2026bonds} are units of debt issued by any person, with mutual credit formed by their exchange, redemption pegging mutually-liquid currencies to each other, financial instruments as swaps of bonds and escrow agreements, and the solvency, liquidity and redemption-run results this paper specialises.  They are undenominated, with no external unit of account.  This paper adds denomination, which makes every coin a debt of a fiat currency, the sovereign, and par by arbitrage.  Their payment protocol is Grassroots Flash~\cite{lewis2023grassroots}, their contract a grassroots social contract~\cite{shapiro2026formalising} compiled to volition-guarded transactions~\cite{lewis2026volitional}, their programs are written in GLP~\cite{shapiro2025glp} and vGLP~\cite{shapiro2026volition}, and the recovery of a lost key or device is specified and proved in~\cite{eitan2026securing}.

\mypara{Stablecoins}
Fiat-backed stablecoins peg a token to a fiat currency by the issuer's obligation to redeem at par, backed by reserves of cash and short-term government debt~\cite{gorton2023taming}; redemption is off-chain and typically restricted to direct customers, other holders realising the peg by selling at the market rate.  Regulation now mandates the obligation: the GENIUS Act in the United States~\cite{genius2025act} and MiCA in the European Union~\cite{mica2023} require at-par redemption rights and full reserve backing.  The obligation of Definition~\ref{def:fiat-bonds} parallels that statutory right, borne by every issuer rather than by one.  A stablecoin has one issuer, whose pooled reserves back all tokens, and its holder's claim is on those reserves; here every person issues their own currency and bears the obligation for it, and the holder's claim is on the issuer, in fiat, fulfilled in kind and, through chains of liquidity, in the sovereign's base money (Proposition~\ref{prop:peg}).

\mypara{Commercial bank money and free banking}
Commercial bank deposits are privately issued money pegged at par to fiat by the bank's obligation to redeem in cash on demand~\cite{mcleay2014money}, and in modern systems by settlement in central bank money, prudential regulation, deposit insurance and the lender of last resort beside it~\cite{garratt2023singleness,bagehot1873lombard}; the note-issuing banks of the free-banking episodes had the obligation and no such backstops behind their notes~\cite{white1984free,selgin1988theory}.  A denomination whose sovereign has issued its grassroots coins reproduces the deposit-to-cash peg without banks: the sovereign's coins take the place of cash, a sovereign-connected person's coins the place of deposits, and the redemption chain the place of the teller (Section~\ref{sec:peg}).  Grassroots coins are inside money~\cite{cavalcanti1999inside,cavalcanti1999model}; the sovereign's coins are the outside money grassroots currencies lacked; and in the framework of digital currency competition~\cite{brunnermeier2019digitalization} such a denomination is a currency area whose unit of account is the fiat unit, joined by holding, issuing or accepting its coins rather than by treaty~\cite{mundell1961theory,alesina2002optimal}.

\mypara{Bills of exchange and promissory notes}
Before banks dominated payment, debts circulated as bills of exchange and promissory notes --- signed promises to pay a fixed sum of a currency at a set date --- traded and discounted before maturity~\cite{rogers1995early,santarosa2015financing}; a grassroots bond is such an instrument in digital form~\cite{shapiro2026bonds}.  Denomination fixes the sum promised in a unit of a fiat currency, which those instruments had and grassroots bonds lacked (Definition~\ref{def:fiat-bonds}).

\mypara{Electronic records in law}
Electronic-signature laws~\cite{esign2000,eidas2014,uncitral2001esignatures} make a signed record legally binding also in electronic form, the UNCITRAL Model Law on Electronic Transferable Records~\cite{uncitral2017mletr} does the same for transferable documents such as bills and notes, and the bills and notes acts~\cite{boe1882,ucc3,geneva1930} govern negotiable instruments.  What follows from them for a grassroots coin, and what this paper does not claim, is in Section~\ref{sec:payment}.

\mypara{Financial networks}
Eisenberg and Noe~\cite{eisenberg2001systemic} model an economy of interbank obligations and prove that a clearing payment vector exists as a fixed point, grounding the analysis of contagion and systemic risk~\cite{glasserman2016contagion}; Karlan, Mobius, Rosenblat and Szeidl~\cite{karlan2009trust} define network-based trust as a maximum flow of social collateral over trust links, and in credit networks a payment goes through exactly when the maximum credit flow from payer to payee covers it~\cite{dandekar2011liquidity}.  How grassroots clearing differs --- unit bonds redeemed in kind at the holder's initiative rather than paid pro rata in a common num\'eraire, a comparison of counts person by person, and a maximum flow over bonds already outstanding rather than over capacities to borrow or to route --- is set out in~\cite{shapiro2024gc}.  Under denomination the counts and the flow are in the fiat unit and the designated member is the sovereign, so what a group can realise now is base money (Section~\ref{sec:runs}).

\mypara{Monetary transmission}
Olivei and Tenreyro~\cite{olivei2007timing} show that the output effect of a monetary policy shock depends on its timing within the year, as staggered wage-setting absorbs shocks arriving just before contracts reset, and Tenreyro and Thwaites~\cite{tenreyro2016pushing} that the effects on output and inflation are much smaller in recessions.  Both concern one economy-wide instrument whose incidence is set by the economy's structure.  Here the incidence is a contract: a term credit line names its counterparty, amount, rate and timing, and the lender of last resort lends to the distressed issuer directly (Section~\ref{sec:policy}), while the policy rate remains, as the rate at which the central bank lends its coins against bonds.  The architecture thereby bypasses intermediary allocation for the named recipient; the aggregate effect on output and prices is a separate question, which this paper does not settle.

\mypara{Community and complementary currencies}
Mutual credit systems such as Sardex~\cite{iosifidis2018sardex} and community currencies such as Sarafu~\cite{mattsson2022sarafu} create liquidity from trust within a community, and their units are typically pegged by convention to the national currency, a peg maintained by the operator's policy rather than by redemption.  Personal-currency systems~\cite{fugger2004money,circles-UBI,trustlines} and credit networks~\cite{dandekar2011liquidity} carry liquidity over trust edges as grassroots currencies do.  A denomination whose sovereign has issued its grassroots coins adds to such systems what convention cannot: redemption at face value, terminating at the sovereign issuer of the unit they already price in.

\mypara{Hawala}
Hawala transfers value across borders without moving money: a customer pays a broker locally, a counterpart broker pays the recipient in their own locale, and the brokers along the path carry the resulting obligations on open mutual account, settled later by reverse transfers, trade or netting, at rates they quote; trust and reputation within the brokers' network take the place of contracts and collateral~\cite{maimbo2003informal}.  Grassroots currencies realise the mechanism: brokers are persons with mutual credit, a transfer is a chain payment across mutually-liquid currencies, and settlement among brokers is set-off.  Denomination adds what hawala lacks: the obligations are legally binding debts of the fiat units and are recorded, where hawala's informality and opacity ground the regulatory concern, and exchange across denominations needs no broker pair, as persons exchange directly, by escrow, both legs completing or neither (Appendix~\ref{app:fx}).  Exchange among persons also bypasses the US dollar as the vehicle of foreign exchange~\cite{gopinath2020dominant}, and the correspondent-banking chain of remittances and its fees~\cite{worldbank2025remittance}.

\mypara{Community credit institutions}
Community credit sustains lending among neighbours without courts or physical collateral: credit cooperatives rely on peer monitoring and local sanctions~\cite{banerjee1994neighbor,guinnane2001cooperatives}, group lending on joint liability --- \emph{social collateral}, repayment secured by the borrower's standing with their community~\cite{besley1995social,ghatak1999joint} --- and rotating savings and credit associations on community-based mutual commitment~\cite{besley1993roscas}.  Caselli and Gennaioli~\cite{caselli2013dynastic} model the failure of meritocracy brought about by imperfect contract enforcement, which prevents the transfer of control over productive assets from the untalented rich to the talented poor; credit secured by social collateral and enforced by the protocol reaches borrowers that courts and asset-based lending do not.  How grassroots credit differs, the sanction carried by the record and the liability individual, each person issuing their own currency and bearing their own default, is set out in~\cite{shapiro2026bonds}.  Under denomination such an institution's currency has the fiat unit its books are already kept in, and a claim on it is a debt of that unit (Definition~\ref{def:fiat-bonds}).
\fi

\section{Conclusions}\label{sec:conclusion}

A central bank digital currency built of grassroots currencies has three components and one contract: the sovereign joins the currency the persons operate as one more issuer, and does not need a separate protocol.  Sovereign grassroots coins are money: digital debts of one unit of the fiat currency, issued by the central bank, held and paid by the public, a direct central bank digital currency by themselves.  We check them against the three foundational principles and the fourteen core features of~\cite{bis2020cbdc}: offline payment is not met; several of the others hold by construction or by the protocols the architecture rests on; and the throughput, resilience and ease of use of a deployment, and the central bank's authority to issue, remain to be established (Sections~\ref{sec:money} and~\ref{sec:cbdc}).  The non-sovereign grassroots currencies add credit: the central bank issues its sovereign grassroots Dollar coins by mutual credit lines  against Dollar-denominated grassroots coins by persons of its choice.  Redemption works the other way: redeeming a sovereign coin gives up the coin and takes a fiat coin, so the coins in circulation fall.  Nothing obliges the central bank to mint a new sovereign coin against a bank deposit: a depositor who wants one buys it from a holder, or their bank pays them out of its cash or the coins it already holds.  A flight from deposits therefore moves coins that already exist and creates none, so no holding limit, tier or waterfall is needed to bound the creation of central bank money (Section~\ref{sec:credit}).  Grassroots bonds add maturity and thus interest, the standard financial instruments, and the central bank's instruments of monetary policy (Section~\ref{sec:interest}).  With them the central bank lends, absorbs liquidity, sets its rates and buys and sells securities in the coins and bonds the public holds, choosing the counterparties and terms of its credit operations.

\iffull
The peg has two legs: the issuer's obligation, worth the issuer's credit, and redemption along chains of mutual liquidity into the sovereign's base money (Proposition~\ref{prop:peg}).  The arbitrage-free price of the coins of an issuer who redeems them on demand is then one unit of the fiat currency (Proposition~\ref{prop:par}, Corollary~\ref{cor:par-sovereign}), so while redemptions are met the currencies of the persons and of the sovereign are one currency, with the fiat currency as its unit of account.  The central bank can choose to deal with any counterparty, not just banks, and its interest rates on lending and bonds bound from above and below the corresponding interest rates of its counterparties (Section~\ref{sec:corridor}).
\fi

Each ingredient has a literature (Section~\ref{sec:related-work}).  Here we proposed their conjunction: a retail coin of the central bank, countersigned by it, in the same contract as the fiat-denominated coins of persons; credit created between two persons, with no shared ledger; redemption chains that carry the sovereign peg; and maturity and escrow, which extend the same instruments to lending and to the operations of monetary policy.  The design lacks what the currencies in operation have: measured throughput, deployed scale and a measured ease of use; it rests on a prototype realisation~\cite{lewis2023grassroots,eitan2026securing}.  Sovereign and non-sovereign grassroots coins and bonds have been implemented and tested on a small scale (Appendix~\ref{app:run}).  A digital economy of non-sovereign grassroots currencies can form without the central bank of the denomination, spontaneously or by spilling over from a neighbouring economy, and the central bank can join it once formed by issuing its own denominated grassroots coins, which are then the sovereign ones.

\ifanon\else
\mypara{AI Disclosure}
We used Anthropic's Claude models to assist with the prose of this paper, for drafting and proofreading: we specified the content, reviewed every passage against the sources and against the papers this one builds on, and made all design decisions.  We verified the correctness and originality of all content, including all references.  Responsibility for all content rests with us.
\fi

\bibliography{bib}

\newpage
\appendix

\section{The Act Schemas, Transactions and Modalities of the Contract}\label{app:schemas}

A grassroots social contract~\cite{shapiro2026formalising} is a finite list of clauses its parties undertake towards one another; the contract of denominated grassroots currencies is stated in Section~\ref{sec:denominated}.  A clause that describes an act is written as an \emph{act schema}: the act's roles, which of them must will it, and the atoms it requires, forbids, adds and deletes at each role.  The volition-guarded multiagent atomic transactions~\cite{lewis2026volitional} realising the contract are compiled from the schemas; three syntactic conditions on the schemas suffice for the protocol so realised to be volitionally grassroots; and each clause is \emph{enforced}, \emph{attested} or \emph{undertaken} by the contract according to what its runs determine and record.  This appendix writes the clauses as act schemas, recalls the transactions of Definition~\ref{def:sgc-transactions} and maps each to its clause, identifies the modality of each clause, and shows the contract syntactically grassroots.

Clauses 1--6 and 9 describe acts; 7 and 8 describe none.  The two clauses new here are of different kinds: clause~8 denominates, adding the variable $f$ to every act and adding no act, while clause~4 adds the acts accept, payout and spend.  Clause~4 is the extension of the contract of grassroots currencies by which a party undertakes to accept named currencies besides their own~\cite{shapiro2024gc}, and with it clause~5 pays coins of the payee, of the payer, or of a third party, who takes part in the payment.  No clause names the central bank, which as a party is bound by clauses 1--9 like any other.  The denomination of a debt is not its governing law: which law governs an obligation, and in which forum it is enforced, are settled outside the contract.  A swap across denominations is not an act of the contract: it is two swaps, one in each denomination, or an escrow whose transfers are each in one (Appendix~\ref{app:fx}).

\subsection{The Act Schemas}\label{app:sgc-schemas}

Four predicates.  \textcent, of arity three: \textcent$(u,f,d)$ in the state of $p$ is an $f$-denominated bond issued by $u$ with maturity date $d$, held by $p$, and \textcent$(u,f,\bot)$ is an $f$-denominated coin of $u$, due now; $d$ ranges over $\calN \cup \{\bot\}$.  $\mathit{date}$, of arity one: $\mathit{date}(t)$ in the state of $p$ is the date of $p$.  $\mathit{dated}$, of arity two: $\mathit{dated}(u,t)$ in the state of $p$ is $u$'s warranty to $p$ that their date was $t$.  $\mathit{accepts}$, of arity two: $\mathit{accepts}(u,f)$ in the state of $p$ is $p$'s acceptance of $u$'s $f$-denominated coins in payment.  The roles are $p$, $q$, $e$ and $s$, the last the issuer of the coins paid in a payment in the coins of a third party; the party variables are $u$, $v$ and $r$; $f$ is the denomination variable, over the countable set $\calF$, disjoint from $\Pi$; the numeric variables are $d$, $d'$ for maturity dates, $t$, $t'$ for dates, $k$, $k'$ for amounts, and $L$ for an acceptance limit.  The seven clauses that describe acts become twelve schemas, payment taking three, redemption two and escrow three.  Nine are the schemas of the undenominated contract~\cite{shapiro2026bonds} with the denomination added --- \textcent\ takes the denomination as a middle argument, and every schema carries the one denomination variable $f$ --- and three are new here: accept, and the two payments in coins the payee did not issue, payout by their issuer and spend by a holder, with the issuer as a third role.

\begin{center}
\footnotesize
\setlength{\tabcolsep}{3pt}
\resizebox{\textwidth}{!}{%
\begin{tabular}{lll}
\hline
\emph{schema} & \emph{at its first role} & \emph{at its second role} \\
\hline
$\mathit{mint}(p?;\,f,d,k,t) \mid d = \bot \lor d > t :$ & $\mathit{date}(t),\ +k\cdot\text{\textcent}(p,f,d)$ & \\
$\mathit{advance}(p;\,t,t') \mid t' > t :$ & $-\mathit{date}(t),\ +\mathit{date}(t')$ & \\
$\mathit{swap}(p?,q?;\,f,u,d,k,v,d',k') :$ & $-k\cdot\text{\textcent}(u,f,d),\ +k'\cdot\text{\textcent}(v,f,d')$ & $-k'\cdot\text{\textcent}(v,f,d'),\ +k\cdot\text{\textcent}(u,f,d)$ \\
$\mathit{accept}(p?;\,f,u,d,L) :$ & $\text{\textcent}(u,f,d),\ +L\cdot\mathit{accepts}(u,f)$ & \\
$\mathit{pay}(p?,q;\,f,k) :$ & $-k\cdot\text{\textcent}(q,f,\bot)$ & $+k\cdot\text{\textcent}(q,f,\bot)$ \\
$\mathit{payout}(p?,q;\,f,k) :$ & $-k\cdot\text{\textcent}(p,f,\bot)$ & $-k\cdot\mathit{accepts}(p,f),\ +k\cdot\text{\textcent}(p,f,\bot)$ \\
$\mathit{spend}(p?,q,s;\,f,k) :$ & $-k\cdot\text{\textcent}(s,f,\bot)$ & $-k\cdot\mathit{accepts}(s,f),\ +k\cdot\text{\textcent}(s,f,\bot)$ \\
$\mathit{redeem}(p?,q;\,f,r,d') :$ & $-\text{\textcent}(q,f,\bot),\ +\text{\textcent}(r,f,d')$ & $-\text{\textcent}(r,f,d'),\ +\text{\textcent}(q,f,\bot)$ \\
$\mathit{mature}(p?,q;\,f,d,t) \mid d \le t :$ & $-\text{\textcent}(q,f,d),\ +\text{\textcent}(q,f,\bot),\ +\mathit{dated}(q,t)$ & $\mathit{date}(t)$ \\
$\mathit{deposit}(p?,e?;\,f,u,d,k) :$ & $-k\cdot\text{\textcent}(u,f,d)$ & $+k\cdot\text{\textcent}(u,f,d)$ \\
$\mathit{release}(e?,q?;\,f,u,d,k) :$ & $-k\cdot\text{\textcent}(u,f,d)$ & $+k\cdot\text{\textcent}(u,f,d)$ \\
$\mathit{return}(e?,p?;\,f,u,d,k) :$ & $-k\cdot\text{\textcent}(u,f,d)$ & $+k\cdot\text{\textcent}(u,f,d)$ \\
\hline
\end{tabular}}
\end{center}

\noindent ``Dollar-denominated coins and bonds'' is the one variable $f$ of every schema, which every atom of the schema carries: no binding of any schema moves instruments of two denominations, which is Corollary~\ref{cor:no-cross-arbitrage}.  ``If they hold a Carol-coin or a Carol-bond'' of clause~4 is that accept requires at $p$ an instrument of the issuer accepted, of the denomination accepted, and ``up to a limit of their choosing'' is that accept adds $L$ acceptance atoms at $p$, one of which each coin so paid to $p$ deletes.  ``Dollar-denominated coins they hold that Bob accepts'' of clause~5 is the three schemas of payment: pay, in which the payee $q$ is the issuer, whose acceptance is clause~7; payout, in which the payer $p$ is the issuer, and $q$ deletes $k\cdot\mathit{accepts}(p,f)$; and spend, in which the issuer is the third role $s$, requiring nothing, and $q$ deletes $k\cdot\mathit{accepts}(s,f)$.  No payment reads or warrants a date: a coin carries none, and the issuer takes part to countersign the payment and to record it in their log.  The reading of the clauses' other quoted phrases --- the guarding roles, the amounts and maturities of mint, the issuer-judged maturity and the date warranty of mature, the unconstrained $r$ and $d'$ of redeem, the bindings of the mutual credit line and the loan, and the two receivers of release and return, fixed by the escrow agreement (Appendix~\ref{app:escrow}) --- is that of the undenominated schemas~\cite{shapiro2026bonds}, with $f$ carried through.

\subsection{The Transactions}\label{app:transactions}

The clauses compile to volition-guarded multiagent atomic transactions~\cite{lewis2026volitional,shapiro2026formalising}, which we recall from~\cite{shapiro2024gc,shapiro2026bonds} with the denomination and the acceptances added.  Each agent is a person operating a machine, and the machine state of a person $p$ is their holdings $c_p$, a multiset of coins and bonds, initially empty; their date $d_p^*$, initially $0$; the warranties $w_p$ they have received, a set of pairs of an issuer and a date, initially empty; and the acceptances $A_p$ they have granted, a multiset of pairs of an issuer and a denomination, initially empty.  A \emph{configuration} $c$ records the state of every person.  A transaction changes the machine states of its participants and no one else's, is enabled when its precondition holds, and is \emph{guarded} by the persons among its participants whose consent it takes: it is taken only when every guard is willing.  A run is \emph{safe} if every step is a transaction of the protocol, \emph{live} if a class of equivalent transactions that is enabled and willed by its guards is eventually taken, and \emph{correct} if it is both; a protocol is \emph{volitionally grassroots} if any group of persons may operate it on its own, and two groups become connected only by a transaction willed by a member of each~\cite{lewis2026volitional}.  Multiset sum and difference are written $\uplus$ and $\setminus$.

\begin{definition}[The Transactions of Denominated Grassroots Currencies]\label{def:sgc-transactions}
For persons $p \ne q \in \Pi$, $e \in \Pi$ distinct from $p$ and $q$, $u, v, r \in \Pi$, a denomination $f \in \calF$, maturity dates $d, d' \in \calN$, a date $t \in \calN$, amounts $k, k' \ge 1$, and a limit $L \ge 1$, the \temph{transactions of denominated grassroots currencies} are the volition-guarded transactions $(c \rightarrow c', Q')$ below, over the participants named, in which $c$ satisfies the precondition stated and every component of a participant's state not named is unchanged.  We write \textcent$^k_{f,u,\ast}$ for $k$ instruments of $u$ in $f$ of one kind and maturity, coins or bonds, where the transaction does not distinguish them:
\begin{enumerate}
    \item \temph{Mint$_p(f,k)$} and \temph{Mint$_p(f,d,k)$}, over $\{p\}$, guarded by $\{p\}$: the second with precondition $d > d^*_p$, so no bond is born due; $c'_p = c_p \uplus \text{\textcent}^k_{f,p}$ and $c'_p = c_p \uplus \text{\textcent}^k_{f,p,d}$ respectively.
    \item \temph{Advance$_p(t)$}, over $\{p\}$, guarded by $\emptyset$: precondition $t > d^*_p$; $d'^*_p = t$.
    \item \temph{Swap$_{p,q}(\text{\textcent}^k_{f,u,\ast}, \text{\textcent}^{k'}_{f,v,\ast})$}, over $\{p,q\}$, guarded by $\{p,q\}$: precondition $\text{\textcent}^k_{f,u,\ast} \subseteq c_p$ and $\text{\textcent}^{k'}_{f,v,\ast} \subseteq c_q$; $c'_p = c_p \setminus \text{\textcent}^k_{f,u,\ast} \uplus \text{\textcent}^{k'}_{f,v,\ast}$ and $c'_q = c_q \setminus \text{\textcent}^{k'}_{f,v,\ast} \uplus \text{\textcent}^k_{f,u,\ast}$.
    \item \temph{Accept$_p(f,u,L)$}, over $\{p\}$, guarded by $\{p\}$: precondition $\text{\textcent}_{f,u,\ast} \in c_p$; $A'_p = A_p \uplus L\cdot\{(u,f)\}$.
    \item \temph{Pay$_{p,q}(\text{\textcent}^k_{f,u})$}, over $\{p,q,u\}$, guarded by $\{p\}$: precondition $\text{\textcent}^k_{f,u} \subseteq c_p$ and $u = q$ or $k\cdot\{(u,f)\} \subseteq A_q$; $c'_p = c_p \setminus \text{\textcent}^k_{f,u}$, $c'_q = c_q \uplus \text{\textcent}^k_{f,u}$, and $A'_q = A_q \setminus k\cdot\{(u,f)\}$ if $u \ne q$.  Only coins are paid.
    \item \temph{Redeem$_{p,q}(f;\,r,\ast)$}, over $\{p,q\}$, guarded by $\{p\}$: precondition $\text{\textcent}_{f,q} \in c_p$ and $\text{\textcent}_{f,r,\ast} \in c_q$; $c'_p = c_p \setminus \text{\textcent}_{f,q} \uplus \text{\textcent}_{f,r,\ast}$ and $c'_q = c_q \setminus \text{\textcent}_{f,r,\ast} \uplus \text{\textcent}_{f,q}$.  The bearer takes a coin or a bond, at their own choice.
    \item \temph{Mature$_{p,q}(f,d)$}, over $\{p,q\}$, guarded by $\{p\}$: precondition $\text{\textcent}_{f,q,d} \in c_p$ and $d \le d^*_q$; $c'_p = c_p \setminus \text{\textcent}_{f,q,d} \uplus \text{\textcent}_{f,q}$ and $w'_p = w_p \cup \{(q,d^*_q)\}$.  The bearer of a matured bond takes a coin of its issuer, at their own election.
    \item \temph{Deposit$_{p,e}(\text{\textcent}^k_{f,u,\ast})$}, over $\{p,e\}$, guarded by $\{p,e\}$: precondition $\text{\textcent}^k_{f,u,\ast} \subseteq c_p$; $c'_p = c_p \setminus \text{\textcent}^k_{f,u,\ast}$ and $c'_e = c_e \uplus \text{\textcent}^k_{f,u,\ast}$.
    \item \temph{Release$_{e,q}(\text{\textcent}^k_{f,u,\ast})$} and \temph{Return$_{e,p}(\text{\textcent}^k_{f,u,\ast})$}, over $\{e,q\}$ and $\{e,p\}$, guarded by both participants: precondition $\text{\textcent}^k_{f,u,\ast} \subseteq c_e$; the instruments are deleted at $e$ and added at the other participant.
\end{enumerate}
Two transactions are \temph{equivalent} iff they have the same name and arguments, differing only in the configurations $c$ and $c'$.
\end{definition}

Each is a schema of Section~\ref{app:sgc-schemas} at a binding: the participants are the roles, the guards are the guarding roles, the precondition collects what the roles require, and the effect collects what they add and delete; Pay is three schemas, with the issuer $u$ the payee, the payer, or a third role.  Mint, accept, pay, redeem and mature are guarded by the initiator, a swap and each escrow transfer by both parties, and advance-date by no one, as local time advances mechanically; the issuer of the coins paid is a participant of every payment in their coins and not a guard of it.  Denomination adds the argument $f$ to every transaction of grassroots bonds and changes no guard.  Acceptance is a quota and not a cap on exposure: Accept$_p(f,u,L)$ grants $L$, each $u$-coin paid to $p$ uses one, and $p$ may grant again, so $p$ consents in advance, and in a bounded quantity, to every rise in their exposure to $u$; the per-issuer limit of grassroots currencies~\cite{shapiro2024gc} bounds the exposure itself.

The map from clause to transaction: clause~1 is mint, for a coin and for a bond; clause~2 is advance, with the date warranty given in mature; clause~3 is swap; clause~4 is accept; clause~5 is pay, payout and spend; clause~6 is redeem, for a coin presented, and mature, for a matured bond; clause~9 is deposit, release and return; clauses~7 and~8 describe no act and are undertaken (Section~\ref{app:modalities}).

\begin{lemma}[Conservation of Debt~\cite{lewis2026volitional,shapiro2024gc,shapiro2026bonds}]\label{lem:conservation}
In any safe run, the $f$-denominated debts of $p$ in any configuration are exactly those $p$ minted in the prefix of the run ending in it, each of them a $p$-coin or a $p$-bond of one unit of $f$.
\end{lemma}
\begin{proof}
Mint adds new debts of $p$ at $p$; swap, pay, redeem, deposit, release and return move instruments between participants without creating or destroying them; mature replaces a matured $p$-bond by a $p$-coin, the same debtor and the same face value; advance and accept change no holding.
\end{proof}

\begin{corollary}[No Cross-Denomination Transactions]\label{cor:no-cross-arbitrage}
Every transaction of an $f$-denominated grassroots currency is over $f$-denominated coins and bonds.
\end{corollary}
\iffull
\begin{proof}
Every act schema of the contract carries one denomination variable, which every atom of the schema carries (Appendix~\ref{app:schemas}), so no binding of any schema moves instruments of two denominations.
\end{proof}
\fi

\subsection{The Modality of Each Clause}\label{app:modalities}

The definitions are~\cite{shapiro2026formalising}'s.  A \emph{clause} of a contract is a proposition about the conduct of its parties, kept or breached by that conduct, and a party it names as answerable is its \emph{obligor}; the runs of the protocol realising the contract \emph{determine} the clause if whether it is kept is a function of the run, and a transaction \emph{may breach} it if it occurs in some correct run that breaches it.  A transaction is \emph{witnessed against} a party if after it another of its participants holds an atom naming that party.  A clause with obligor $o$ is \emph{enforced} if the runs determine it and every correct run keeps it; \emph{attested} if some transaction may breach it and every transaction that may is witnessed against $o$; and \emph{undertaken} if no transaction may breach it, the runs do not determine it, and some transaction is witnessed against $o$ or $o$'s undertaking is held by another party.

Clauses 1, 3, 4, 5, 6 and 9 are enforced, clause~2 is attested, and clause~7 is undertaken, as for the undenominated contract~\cite{shapiro2026bonds}: an enforced clause permits an act and obliges nothing beyond it, with the obligations of clauses 5 and~6 self-executing in the schemas --- each schema of payment is guarded by the payer and adds the coins at the payee, who accepts them under clause~7 or under clause~4, and redeem is guarded by the bearer and deletes the instrument taken at the issuer; the issuer of the coins spent is a role of spend and not a guard, so they take part and are not asked; the date warranty of clause~2 is given in the course of mature, which adds $\mathit{dated}(q,t)$ in the state of the bearer, witnessed against the issuer $q$ and legally enforceable on that record --- no other schema warrants a date, so mature is the one act that may breach the clause; and the coin a swap adds in another party's holding is the record of the pricing undertaking of clause~7.

Clause~8 is undertaken.  Handing over a unit of a fiat currency is not an act of the contract, so no transaction of it may breach the clause and no run determines whether it is kept; and it is witnessed against the issuer: the coin \textcent$(p,f,d)$ in another party's holding is the record of a debt of one unit of $f$, enforceable under the law governing the obligation (Section~\ref{sec:denominated}).  The in-kind fulfilment of clause~8 is redeem, enforced, and its fiat fulfilment is performed outside the contract (Section~\ref{sec:payment}).  No act of the contract asks for a unit of $f$ either: the presentation for a fiat coin and its payment are both outside it, on the record the coin is; what the contract carries is the coin's return to its issuer, a payment to it.

The sovereign is not a role.  Act schemas cannot name a person, and every party begins in the same empty state, so no schema and no clause distinguishes the sovereign of $f$ from any other party~\cite{shapiro2026formalising}: the map $\sigma$ and the sovereign's ability to issue fiat coins (Section~\ref{sec:the-sovereign}), which set it apart, lie outside the contract.  The sovereign mints, swaps, accepts, pays and redeems under the same twelve schemas, and its warranty under clause~2 and its undertakings under clauses 7 and 8 are recorded against it as against anyone.

\subsection{The Contract Is Grassroots}\label{app:sgc-grassroots}

A contract is syntactically grassroots if it has an unobstructed introductory act, every predicate its acts require has traceable provenance, and every act not willed by all its parties runs only between parties that already record each other --- its role graph, with an edge between two roles when one requires an atom of traceable provenance naming the other, is connected; the protocol realising a syntactically grassroots contract is volitionally grassroots~\cite{shapiro2026formalising}.  The three conditions are checked for the undenominated schemas in~\cite{shapiro2026bonds}, and the argument $f$ is a denomination and names no person, so the check carries to the schemas shared with it unchanged: the swap at the binding of the mutual credit line is the unobstructed introductory act; \textcent\ and $\mathit{dated}$ have traceable provenance, and $\mathit{date}$ names no person; and pay, redeem and mature, guarded by $p$, require at $p$ a coin or a bond of the counterparty, an atom of traceable provenance naming them.  Of the three schemas new here, accept adds $\mathit{accepts}(u,f)$ at $p$ where $u$ is an argument of the instrument $p$ requires, so $\mathit{accepts}$ has traceable provenance; payout, guarded by $p$, has $q$ require $\mathit{accepts}(p,f)$, naming $p$; and spend, guarded by $p$, has $p$ require a coin of $s$ and $q$ require $\mathit{accepts}(s,f)$, both naming $s$, so its role graph is connected through the issuer.  A payment in the coins of a third party between two parties who do not record each other is therefore an act of the contract only with the issuer as a party, which is the countersignature of Section~\ref{sec:payment}.

The contract is therefore syntactically grassroots at every value of its numeric variables.  Its schemas carry numeric variables and the denomination variable, each over a countable set, so their instances are a countable set, and the theorem of~\cite{shapiro2026formalising} is for a finite contract; the conclusion for the countable set is through politeness, a condition on transactions and not on schemas, as for grassroots currencies~\cite{shapiro2024gc}.

\begin{lemma}[Amounts, Dates and Denominations]\label{lem:amounts}
The transactions of Definition~\ref{def:sgc-transactions} are a polite set~\cite{shapiro2026formalising}, and the protocol over them is volitionally grassroots.
\end{lemma}

\begin{proof}
A set of volition-guarded transactions is \emph{polite} if it is open --- any two persons can always be brought to a configuration at which some transaction coupling them is enabled --- and closed --- every transaction enabled at the end of a safe run has all its participants in one instance of that run, or a guard of it in every instance holding one of them; the protocol over a polite set is volitionally grassroots~\cite{shapiro2026formalising}.  Both are conditions on the set of transactions, and neither depends on the number of schemas.  The checks above are per schema and per binding and hold at every value of the numeric variables, so each instance is a contract of~\cite{shapiro2026formalising} that is syntactically grassroots.  Closure quantifies over the runs of the protocol over the whole set, and a run of the union is not a run of a single instance, so it is proved here of the union, as for grassroots currencies~\cite{shapiro2024gc}.  \textcent, $\mathit{dated}$ and $\mathit{accepts}$ have traceable provenance in every schema of every instance, and the induction on a run by which~\cite{shapiro2026formalising} proves that no transition puts an edge whose ends lie in different instances uses only the schema inducing the current transition, so it carries to the union unchanged: whenever $p$ holds an atom naming $q$, $p$ and $q$ lie in one instance of the run.  Mint, advance and accept are unary; swap, deposit, release and return are guarded in all their roles; a pay or a redeem enabled between $p$ and $q$ requires $p$ to hold a $q$-coin, a mature enabled between $p$ and $q$ requires $p$ to hold a $q$-bond, a payout enabled between $p$ and $q$ requires $q$ to hold $\mathit{accepts}(p,f)$, and a spend enabled over $p$, $q$ and $s$ requires $p$ to hold an $s$-coin and $q$ to hold $\mathit{accepts}(s,f)$, so the participants of each already lie in one instance.  The union is therefore closed.  For openness take any two persons $p \ne q$: each mints one $f$-denominated coin of their own, by unary transactions, and a swap of the two is then enabled, and it is a transaction of the union.  The union is therefore polite, and the protocol over it is volitionally grassroots.
\end{proof}

\noindent The contract is therefore syntactically grassroots, and the protocol realising it is volitionally grassroots: any group of persons operates it on its own, and two groups become connected only by a swap willed by a member of each.

\section{Financial Instruments}\label{app:escrow}

This appendix presents the financial instruments of grassroots bonds~\cite{shapiro2026bonds}, recalled here so that the claim of Section~\ref{sec:instruments} rests on the constructions and not solely on the reference.  Every transaction of an instrument is over coins and bonds of one denomination (Corollary~\ref{cor:no-cross-arbitrage}), which we leave implicit below; \textcent$^k_{p,d}$ is a multiset of $k$ bonds issued by $p$ and maturing at $d$, and \textcent$^k_p$ is a multiset of $k$ $p$-coins.  Every transaction below is a transaction of Definition~\ref{def:sgc-transactions}, so no instrument adds a construct to the currency; what an instrument adds is the escrow program that fixes what the agent consents to.

\subsection{Instruments That Are One Swap}\label{app:swaps}

A swap exchanges one lot --- coins of one issuer, or bonds of one issuer and maturity --- for another, by an act of both parties.  Each instrument below is a swap between $p$ and $q$, specified by the instruments $x$ that $p$ gives and the instruments $y$ that $q$ gives.

\begin{enumerate}
\item \textbf{Symmetric mutual credit.}  $x = \text{\textcent}^k_p$, $y = \text{\textcent}^k_q$: each mints and gives the same number of their own coins, and each thereafter holds coins of the other~\cite{shapiro2024gc}.
\item \textbf{Asymmetric mutual credit.}  $x = \text{\textcent}^{k'}_p$, $y = \text{\textcent}^k_q$ with $k' < k$: $p$ charges a premium for credit to $q$.  As coins are redeemable at par, $p$ may redeem at once and keep the premium, so a premium on coins is not interest~\cite{shapiro2024gc}.
\item \textbf{Sale of debt.}  $x = \text{\textcent}^k_r$, $y = \text{\textcent}^{k'}_q$ with $k' < k$: $p$ sells coins of a third person $r$ to $q$ at a discount, which prices $p$'s doubt about $r$'s solvency.  The debtor $r$ is not a party.
\item \textbf{Zero-coupon loan.}  $x = \text{\textcent}^{k'}_p$, $y = \text{\textcent}^k_{q,d}$ with $k' < k$: the lender $p$ gives coins now, the borrower $q$ bonds of larger face value maturing at $d$, and the difference is interest for that term.  A bond is not presented for redemption, so the loan cannot be recalled; the borrower may repay early, by redeeming the lender's coins against its own bond.
\item \textbf{Sale of debt before maturity.}  $x = \text{\textcent}^k_{r,d}$, $y = \text{\textcent}^{k'}_q$ with $k' < k$: the discount prices the time to maturity as well as the doubt about $r$.
\item \textbf{Forward contract.}  $x = \text{\textcent}^k_{p,d}$, $y = \text{\textcent}^{k'}_{q,d}$: bonds of the two parties maturing at one date are swapped at inception, so at $d$ each holds a claim on the other of the agreed face value.  Between inception and $d$ each party holds the other's debt, so a default in that interval falls on the counterparty, which a forward settled at $d$ does not.
\end{enumerate}

A currency swap and a repurchase agreement are sequences of two swaps; factoring is a sale of debt; a mortgage is a loan with collateral; a bank's time deposit is its bond, and its demand deposit is its coin.

\subsection{The Escrow Agent}\label{app:escrow-agent}

An escrow agent $e$ holds coins and bonds on behalf of the parties to an escrow agreement and transfers them by its agreed conditions.  Each transfer is a transaction between $e$ and the party giving or receiving them, and $e$ is party to every transfer, so each is made by the consent of $e$ and that party.  What $e$ consents to is fixed by the escrow program implementing the agreement: for a date-based condition, by the local date $d_e^*$ of $e$; for an event-based condition, by the judgment of $e$, informed where needed by an \emph{oracle}, an agent the parties agree to trust for the value of a quantity outside the agreement, which attests but holds no instruments and takes no act.  An attestation is not a transaction of the contract: the escrow agent judges from what it sees, its own date and the deposits it holds, and an outside event reaches it outside the contract, as the payment of a fiat coin does; the instruments below that rest on an outside event rest on that.  The constructions assume that both parties trust $e$ to follow the escrow program and, for event-based conditions, trust $e$ as adjudicator of the condition.  An arrangement has a depositor, a beneficiary and $e$: the depositor transfers coins or bonds to $e$; subsequently $e$ either releases them to the beneficiary or returns them to the depositor.  We write $\{(a, x),\; (b, y)\}$ for the exchange in which $a$ gives the instruments $x$ and $b$ gives the instruments $y$, guarded by $a$ and $b$; $y = \emptyset$ makes it a one-sided transfer.  The contract has four such transfers: the deposit, the release and the return of clause~9, and the payment of clause~5, which carries only coins the party receiving them accepts.  Bonds therefore do not pass one way from one ordinary party to another; a swap has two sides.

An escrow whose agent is a committee, acting when a threshold of its members agree, is not expressible in a model in which every agent is one person~\cite{shapiro2026bonds}.

\subsection{Instruments That Employ an Escrow Agent}\label{app:escrow-instruments}

\begin{enumerate}
\item \textbf{Packaged exchange.}  Parties $p$ and $q$ agree the lots $X$ that $p$ gives, the lots $Y$ that $q$ gives, and an establishment window $T_0$.  \emph{Deposit}: $\{(p, X),\; (e, \emptyset)\}$ and $\{(q, Y),\; (e, \emptyset)\}$, one transfer per lot.  \emph{Release}: once every lot of both is deposited, $\{(e, Y),\; (p, \emptyset)\}$ and $\{(e, X),\; (q, \emptyset)\}$.  \emph{Return}: if the window passes with the lots of a party incomplete, $e$ returns every deposited lot to its depositor.  All lots change hands or none do, and neither party holds the other's lots before both deposits are complete.  It realises the instruments whose sides are several lots:
\begin{enumerate}
\item \textbf{Balloon loan.}  $X = \text{\textcent}^k_p$; $Y = \bigcup_{j=1}^{n} \text{\textcent}^{k_j}_{q,d_j} \cup \text{\textcent}^k_{q,d}$: coupon bonds at the payment dates and a principal bond at $d$.
\item \textbf{Fixed-payment loan.}  $X = \text{\textcent}^k_p$; $Y = \bigcup_{j=1}^{n} \text{\textcent}^{k_j}_{q,d_j}$: bonds with fixed payments $k_j$ at dates $d_j$.
\item \textbf{Interest rate swap.}  $X = \bigcup_j \text{\textcent}^{k_j}_{p,d_j}$, $Y = \bigcup_j \text{\textcent}^{k'_j}_{q,d_j}$, where $p$ pays fixed amounts $k_j$ while $q$ pays amounts $k'_j$ that vary with a reference rate, attested by an interest-rate oracle.  Periodic settlement exchanges realise the net difference at each date.
\end{enumerate}

\item \textbf{Term credit line.}  A lender $p$ commits to advance up to $k$ coins to a borrower $q$, with rate $\rho$ on drawn credit, implied by the face values, payment dates $d_1, \ldots, d_n$ and expiry $T$.  \emph{Establishment}: $p$ mints $k$ coins and deposits them with $e$, $\{(p, \text{\textcent}^k_p),\; (e, \emptyset)\}$; $e$ tracks the drawn amount $k_d$, initially $0$.  \emph{Draw}: when $q$ draws $k' \le k - k_d$, $q$ mints a principal bond \textcent$^{k'}_{q,T}$ and coupon bonds \textcent$^{\rho k'}_{q,d_j}$, with $\rho k'$ an integer in the smallest unit (Definition~\ref{def:fiat-bonds}), for each $d_j > d_e^*$ and deposits them with $e$; $e$ releases them to $p$ and releases \textcent$^{k'}_p$ to $q$, and $k_d$ becomes $k_d + k'$.  \emph{Repayment}: $q$ acquires $p$-coins and redeems them from $p$ against its principal bond, which $p$ cannot refuse; $p$ deposits them with $e$, and $k_d$ falls by that amount.  The coupon bonds of a repaid draw stand, being $q$'s debts already issued and held by $p$.  \emph{Expiry}: once $d_e^* \ge T$, $e$ returns the undrawn balance to $p$.  Between the operations of the agreement $e$ holds \textcent$^{k-k_d}_p$, and $p$ holds \textcent$^{k_d}_{q,T}$ together with the coupon bonds of the draws, with $0 \le k_d \le k$, so $q$ cannot draw beyond the limit and the principal $p$ holds equals the amount drawn~\cite{shapiro2026bonds}.

\item \textbf{Collateral.}  To collateralise a loan from $p$ to $q$: \emph{deposit} $\{(q, x),\; (e, \emptyset)\}$, the borrower's bonds $x$ to $e$; \emph{release on default} $\{(e, x),\; (p, \emptyset)\}$; \emph{return on fulfilment} $\{(e, x),\; (q, \emptyset)\}$.  The lender of last resort of Section~\ref{sec:policy} lends against such a deposit.

\item \textbf{Guarantee.}  A guarantor $g$ covers the obligations of $q$ to $p$: \emph{deposit} $\{(g, \text{\textcent}^k_g),\; (e, \emptyset)\}$; \emph{invocation on default} $\{(e, \text{\textcent}^k_g),\; (p, \emptyset)\}$; \emph{release on fulfilment} $\{(e, \text{\textcent}^k_g),\; (g, \emptyset)\}$.

\item \textbf{Option.}  An option grants $p$ the right, and not the obligation, to execute a specified swap with $q$ within an exercise window $[d_1, d_2]$ agreed at establishment and judged by $d_e^*$, with an establishment window $T_0$ before which both parties must deposit.  \emph{Establishment}: $p$ deposits the premium \textcent$^k_p$ with $e$ and $q$ deposits the underlying bonds $y$; if only one deposits by $T_0$, $e$ returns that deposit.  \emph{Exercise}: while $d_1 \le d_e^* \le d_2$, $p$ may signal exercise, whereupon $e$ releases $y$ to $p$ and the premium to $q$.  \emph{Expiry}: once $d_e^* > d_2$ with no exercise, $e$ returns $y$ to $q$ and transfers the premium to $q$.  American exercise is the window $[0, d]$ and European exercise the window $[d, d_2]$; a window is needed because advance-date is unguarded and may raise $d_e^*$ past $d$ without equalling it.  This is a call option; a put option is symmetric, $q$ depositing the strike price and $p$ the bonds it wishes to sell.  A strike price is incorporated by a further deposit of $p$ at establishment, released to $q$ on exercise and returned to $p$ on expiry.

\item \textbf{Insurance.}  The insured $p$ pays a premium and the insurer $q$ a payout if a specified event occurs by expiry $T$, as $e$ attests.  \emph{Establishment}: $\{(q, \text{\textcent}^{k'}_q),\; (e, \emptyset)\}$, the insurer's reserve, $k' \gg k$.  \emph{Premium}: $\{(p, \text{\textcent}^k_p),\; (e, \emptyset)\}$.  \emph{Claim}: $\{(e, \text{\textcent}^{k'}_q),\; (p, \emptyset)\}$ and $\{(e, \text{\textcent}^k_p),\; (q, \emptyset)\}$.  \emph{Expiry with no claim}: $e$ transfers the premium to $q$ and returns its reserve.  A mutual insurance arrangement is each member depositing premiums into a shared escrow reserve.

\item \textbf{Credit default swap.}  $p$ pays periodic premiums to the protection seller $q$, who compensates $p$ if a reference person $r$ defaults, as $e$ adjudicates.  \emph{Establishment}: $\{(q, \text{\textcent}^K_q),\; (e, \emptyset)\}$, the seller's reserve.  \emph{Premium payments}: $\{(p, \text{\textcent}^{k_j}_{p,d_j}),\; (e, \emptyset)\}$ at the dates $d_j$, forwarded to $q$.  \emph{Credit event}: $\{(e, \text{\textcent}^K_q),\; (p, \emptyset)\}$.  \emph{No credit event}: the premium bonds go to $q$, and at expiry $T$ the reserve returns to $q$.

\item \textbf{Letter of credit.}  A bank $b$ guarantees payment to the seller $p$ on behalf of the buyer $q$, $e$ verifying that the contractual conditions are met.  \emph{Issuance}: $\{(b, \text{\textcent}^k_b),\; (e, \emptyset)\}$.  \emph{Presentation}: once $e$ verifies that $p$ has fulfilled the terms, $\{(e, \text{\textcent}^k_b),\; (p, \emptyset)\}$.  \emph{Reimbursement}: the buyer deposits its bonds with $e$, $\{(q, \text{\textcent}^k_{q,d}),\; (e, \emptyset)\}$, and $e$ releases them to the bank, $\{(e, \text{\textcent}^k_{q,d}),\; (b, \emptyset)\}$; the buyer cannot transfer them to the bank directly, there being no such act.
\end{enumerate}

\section{The Sovereign Market, as It Runs}\label{app:run}

The transactions of Definition~\ref{def:sgc-transactions} are implemented by the \emph{denominated agent} every party runs, a program in vGLP~\cite{shapiro2026volition}, GLP~\cite{shapiro2025glp} with volition guards: a clause preceded by a volition guard \verb|*(X1, ..., Xk)| reduces only when the person wills it, and the person's answer fills the writers \verb|X1, ..., Xk|.  The agent has a volition-guarded request clause for each transaction a party initiates: Mint has two, for coins and for bonds, and Pay has three, to the issuer of the coins paid, by that issuer, and by a third party.  Swap and the three escrow transfers are guarded by both parties, and the other party answers by the volition-guarded clause of a responder.  Advance is guarded by no one: the agent reads its date from a calendar.

The agent is that of grassroots bonds~\cite{shapiro2026bonds} with the denomination added as the argument after the issuer of every coin and bond, and its arguments are those of the agent of grassroots currencies~\cite{shapiro2024gc} with the calendar, the date, the warranties received and the acceptances granted added.  Accept and the payments by the issuer and by a third party are new here, and their request clauses are below.  The central bank runs the same agent, and no clause names it.

\mypara{Accept}
Accept$_p(f,u,L)$ is guarded by $p$: a request clause with three fields, the denomination, the issuer accepted and the limit, whose precondition $\text{\textcent}_{f,u,\ast} \in c_p$ is tested by \verb|holds|, and whose effect adds $L$ acceptances of $u$ in $f$ to $A_p$:
\begin{verbatim}
*(F, U, L)
agent(Id, UserIn, PeerIn, Convs, Cal, Outs, Holdings, Date,
      Ws, As, Serial) :-
    ground(Id?), ground(F?), ground(U?), integer(L?), L? >= 1 |
    holds(U?, F?, Holdings?, Held, Holdings1),
    do_accept(Held?, Id?, U?, F?, L?, UserIn?, PeerIn?,
              Convs?, Cal?, Outs?, Holdings1?, Date?, Ws?,
              As?, Serial?).
\end{verbatim}

\mypara{A payment by the issuer}
Pay$_{p,q}(\text{\textcent}^k_{f,p})$ is guarded by $p$: a request clause with three fields, the payee, the denomination and the amount, whose precondition $\text{\textcent}^k_{f,p} \subseteq c_p$ is tested by selecting $p$'s own coins:
\begin{verbatim}
*(Q, F, K)
agent(Id, UserIn, PeerIn, Convs, Cal, Outs, Holdings, Date,
      Ws, As, Serial) :-
    ground(Id?), ground(Q?), ground(F?), integer(K?), K? >= 1 |
    select_instruments(Id?, F?, now, K?, Holdings?,
                       Status, Give, Rest),
    do_payout(Status?, Id?, Q?, F?, K?, Give?, Rest?, UserIn?,
              PeerIn?, Convs?, Cal?, Outs?, Date?, Ws?, As?,
              Serial?).
\end{verbatim}
The coins reach $q$, whose agent takes them without asking its person if $q$ holds $k$ acceptances of $p$ in $f$, and deletes $k$ of them; otherwise it returns the coins.

\mypara{A payment by a third party}
Pay$_{p,q}(\text{\textcent}^k_{f,u})$ with $u$ distinct from $p$ and $q$ is guarded by $p$: a request clause with four fields, the payee, the issuer, the denomination and the amount, whose precondition $\text{\textcent}^k_{f,u} \subseteq c_p$ is tested by selecting the coins of $u$:
\begin{verbatim}
*(Q, S, F, K)
agent(Id, UserIn, PeerIn, Convs, Cal, Outs, Holdings, Date,
      Ws, As, Serial) :-
    ground(Id?), ground(Q?), ground(S?), ground(F?),
    integer(K?), K? >= 1 |
    select_instruments(S?, F?, now, K?, Holdings?,
                       Status, Give, Rest),
    do_spend(Status?, Id?, Q?, S?, F?, K?, Give?, Rest?,
             UserIn?, PeerIn?, Convs?, Cal?, Outs?, Date?, Ws?,
             As?, Serial?).
\end{verbatim}
The coins reach $u$, whose agent countersigns the payment without asking its person and forwards it to $q$ on the conversation $u$ holds with $q$.  The agent of $q$ takes the coins without asking its person if $q$ holds $k$ acceptances of $u$ in $f$, and deletes $k$ of them; otherwise it returns the coins.  The issuer records the payment when $q$ has taken it.  The payer and the payee need no conversation with each other; the issuer connects them (Appendix~\ref{app:sgc-grassroots}).

\mypara{The village}
The run is the village market of grassroots bonds~\cite{shapiro2026bonds} --- Alice the baker, Bob the farmer, Charlie the carpenter, Diana the doctor, Eve the teacher, Frank the fisherman, and an escrow agent --- denominated in one fiat currency, with the central bank as an eighth party and Diana as the community bank.  Besides the operations of that village, the central bank mints thirty coins and opens a mutual credit line with Diana; Bob redeems two Diana-coins for central bank coins and accepts central bank coins up to five; Frank redeems a Diana-coin for a central bank coin, accepts central bank coins up to five, and pays the coin to the central bank for a unit of the fiat currency; the central bank pays Frank four of its coins; Frank spends three of them at Bob; Diana asks for a Frank-coin for a Frank-bond before Frank's date reaches its maturity, and again after; and Frank deposits two Eve-coins with the escrow agent, who returns them.  A play stands in for the super-app and for the persons: it opens the conversation between every two parties who transact, and a scripted person for each party answers the cards the agent shows in the order of their script.

Below is the output of the run, elided to the transactions this paper rests on.  A holding is a list of lots, each naming its issuer, its denomination, \texttt{now} for a coin or the maturity date of a bond, and its number of instruments, so \texttt{lot(cb,~usd,~now,~30)} is thirty sovereign grassroots coins and \texttt{lot(bob,~usd,~25,~24)} twenty-four Bob-bonds maturing on day~$25$.  An \texttt{opened(q)} line records a conversation with $q$ that the play has opened: two parties transact over the root channel of a conversation and over nothing else.

\begin{footnotesize}
\begin{verbatim}
...
tagged(cb, opened(diana))
...
tagged(cb, opened(frank))
...
tagged(cb, opened(bob))
...
tagged(cb, minted(30, now))
tagged(frank, minted(18, 28))
...
tagged(diana, swap_done(cb))
tagged(cb, swap_done(diana))
tagged(diana, holdings([lot(cb, usd, now, 10)]))
...
tagged(cb, holdings([lot(cb, usd, now, 20), lot(diana, usd, now, 10)]))
...
tagged(bob, redeemed(diana, cb, now))
...
tagged(bob, accepted(cb, usd, 5))
...
tagged(frank, not_mature_for(diana, 28))
tagged(diana, not_mature(frank, 28))
...
tagged(frank, date_advanced(28))
...
tagged(escrow, transferred(release, frank))
...
tagged(diana, matured(frank, 28))
tagged(diana, warranty(frank, 28))
...
tagged(escrow, transferred(return, frank))
...
tagged(frank, redeemed(diana, cb, now))
tagged(frank, holdings([..., lot(cb, usd, now, 1)]))
...
tagged(frank, accepted(cb, usd, 5))
...
tagged(frank, paid(cb, 1))
tagged(cb, received(frank, 1))
tagged(cb, fiat_handed(frank, 1))
tagged(cb, holdings([lot(cb, usd, now, 21), lot(diana, usd, now, 10)]))
...
tagged(cb, paid_out(frank, 4))
tagged(frank, received_from(cb, cb, 4))
tagged(frank, acceptances([accepting(cb, usd, 1)]))
...
tagged(frank, spending(bob, cb, 3))
tagged(frank, holdings([..., lot(cb, usd, now, 1)]))
tagged(bob, received_from(frank, cb, 3))
tagged(bob, acceptances([accepting(cb, usd, 2)]))
tagged(frank, spent_at(bob, cb, 3))
tagged(bob, holdings([..., lot(cb, usd, now, 5)]))
tagged(cb, countersigned(frank, bob, 3))
\end{verbatim}
\end{footnotesize}

\noindent The run ends with the holdings of Table~\ref{table:sovereign-holdings}.  They reconcile with Conservation of Debt (Lemma~\ref{lem:conservation}): the instruments of each issuer held across the eight parties, the issuer's own included, are those they minted --- $30$ central bank coins, $30$ Alice-coins, $15$ Bob-coins and $24$ Bob-bonds, $25$ Charlie-coins, $45$ Diana-coins, $20$ Eve-coins, and Frank's $10$ coins and $18$ bonds, one of which matured and is held as a Frank-coin.

\begin{table}[!htbp]
\centering
\small
\begin{tabular}{@{}ll@{}}
\hline
\textbf{Party} & \textbf{Holdings after the run} \\
\hline
Central bank & 17 cb-coins, 10 diana-coins \\
Alice & 5 bob-coins, 15 charlie-coins, 8 alice-coins, 4 frank-coins \\
Bob & 10 alice-coins, 18 diana-coins, 5 cb-coins \\
Charlie & 10 alice-coins, 10 eve-coins, 6 charlie-coins \\
Diana & 7 cb-coins, 11 diana-coins, 1 frank-coin, \\
 & 24 bob-bonds maturing on day 25, 12 frank-bonds maturing on day 28 \\
Eve & 4 charlie-coins, 1 frank-coin, 2 alice-coins, 10 bob-coins \\
Frank & 6 diana-coins, 10 eve-coins, 5 frank-coins, 1 cb-coin, \\
 & 5 frank-bonds maturing on day 28 \\
Escrow agent & none \\
\hline
\end{tabular}
\caption{The holdings of the eight parties at the end of the sovereign market, as each party's screen reports them, self-held instruments included; every instrument is denominated in the one fiat currency, and cb is the central bank.}
\label{table:sovereign-holdings}
\end{table}

\mypara{What the run shows}
The central bank mints thirty coins and opens a mutual credit line with Diana: a swap of ten of its coins for ten of hers, at par, on demand and without interest, willed by both (Section~\ref{sec:credit-lines}).  Frank, a household, holds Diana-coins and Diana holds central bank coins, so the $f$-coin holding graph has the path $\mathit{frank} \rightarrow \mathit{diana} \rightarrow \mathit{cb}$, and after one redemption along it Frank holds a sovereign grassroots coin, which is Proposition~\ref{prop:peg} performed (Section~\ref{sec:peg}).  Frank pays that coin to the central bank, a payment to its issuer (Definition~\ref{def:sgc-transactions}), and the central bank's person hands over one unit of the fiat currency, which is outside the contract (Appendix~\ref{app:schemas}); the coins of the central bank in circulation fall by one.  The central bank then pays out four of its coins to Frank, who has accepted them (Section~\ref{sec:issuance}).  Frank spends three of them at Bob, who has accepted them, and the central bank countersigns and records the payment: a retail payment in sovereign grassroots coins between two parties who have no conversation open (Section~\ref{sec:payment}).  Diana's first request for a Frank-coin for a Frank-bond is returned, Frank's date not having reached its maturity; after it has, Diana takes a Frank-coin for the bond and holds Frank's warranty of his date, the one warranty given in the run.  Charlie's deposit is released to Frank when the escrow agent's date reaches day~$30$, and Frank's deposit is returned to him (Appendix~\ref{app:escrow}).

\mypara{What the run does not show}
Throughput, deployed scale and ease of use are yet to be measured.

\section{Exchange across Denominations}\label{app:fx}

Persons can exchange grassroots coins across denominations among themselves, with no bank accounts and no correspondent banking.  Coins of any two denominations can be exchanged directly, so the exchange may bypass the vehicle currency of today's markets: the US dollar, the vehicle of foreign exchange and the invoicing currency of commodities~\cite{gopinath2020dominant}.

\mypara{Exchange as paired payments}
A person $p$ holding Dollar-denominated grassroots coins and a person $q$ holding Euro-denominated ones agree a rate and exchange by a pair of payments: $p$ pays $q$ $x$ Dollar-denominated coins acceptable to $q$, and $q$ pays $p$ $y$ Euro-denominated coins acceptable to $p$, each payment a transaction of one denomination (Section~\ref{sec:denominated}).  The rate $y/x$ is set by the parties: the architecture neither performs nor constrains conversion (Corollary~\ref{cor:no-cross-arbitrage}), and rates are set by competition among the persons quoting them.  The two payments are separate transactions, so one party is exposed between the legs; the exposure is bounded by exchanging coin by coin, or by the mutual credit the parties extend each other.

\mypara{Escrowed exchange}
The escrow agent of grassroots bonds (Appendix~\ref{app:escrow}) holds coins and bonds on behalf of the parties and transfers them by its conditions, and every deposit, release and return is a transfer of instruments of one denomination, so escrow composes with denomination.  An exchange agreement: $p$ deposits $x$ Dollar-denominated coins with the escrow agent $e$, $q$ deposits $y$ Euro-denominated ones, and $e$ releases the two deposits crosswise, or returns a lone deposit if the other is not made within the establishment window.  Both legs complete or neither, at the end of the agreement: the parties need no mutual credit and no trust in each other, only in $e$.  The agreement is not one transaction --- no transaction spans two denominations (Corollary~\ref{cor:no-cross-arbitrage}) --- so between the two releases there is a state in which one leg has completed and the other has not; the escrow agent guarantees the outcome of the agreement, and in that state each party's exposure is to $e$.  Forwards and options across denominations are the escrow instruments of~\cite{shapiro2026bonds} with legs in distinct denominations, each transaction within one denomination.

\mypara{Consequences}
A person holding grassroots coins of several denominations, quoting rates and standing ready to exchange, is a money changer, and competition among money changers sets the rates.  A migrant worker can exchange into the home denomination and pay family directly, avoiding the correspondent-banking chain and its fees~\cite{worldbank2025remittance}, though a money changer, an issuer or an escrow agent may charge a spread or a fee.  A person with assets in one denomination and a liability in another covers by exchange, bearing the rate risk.  Exchange, remittance and covering thus proceed among persons, in legally binding debts of the fiat units, without correspondent banking.

\end{document}